\documentclass[12pt,cls,onecolumn]{IEEEtran}
\usepackage{subfigure,cite,graphicx,psfig,amsmath,amssymb,mathrsfs,epsf}

\begin{document}

\title{Codebook-Based Opportunistic Interference Alignment}
\author{\large Hyun Jong Yang, \emph{Member}, \emph{IEEE}, Bang Chul Jung, \emph{Member}, \emph{IEEE}, \\ Won-Yong Shin, \emph{Member}, \emph{IEEE}, and
Arogyaswami Paulraj, \emph{Fellow}, \emph{IEEE} \\
\thanks{The material in this paper was presented in part at the
IEEE International Symposium on Information Theory, Cambridge, MA,
July 2012.}
\thanks{H. J. Yang is with the School of Electrical and Computer Engineering, UNIST, Ulsan 689-798, Republic of Korea (E-mail:
hjyang@unist.ac.kr).}
\thanks{B. C. Jung is with the Department of Information and Communication Engineering, Gyeongsang National
University, Tongyeong 650-160, Republic of Korea (E-mail:
bcjung@gnu.ac.kr).}
\thanks{W.-Y. Shin (corresponding author) is with the Department of Computer Science and
Engineering, Dankook University, Yongin 448-701, Republic of Korea
(E-mail: wyshin@dankook.ac.kr).}
\thanks{A. Paulraj is with the Department of Electrical Engineering,
Stanford University, Stanford, CA 94305 (email:
apaulraj@stanford.edu). }
} \maketitle


\markboth{Submitted to IEEE Transactions on Signal Processing} {Yang
{\em et al.}: Codebook-Based Opportunistic Interference Alignment}


\newtheorem{definition}{Definition}
\newtheorem{theorem}{Theorem}
\newtheorem{lemma}{Lemma}
\newtheorem{example}{Example}
\newtheorem{corollary}{Corollary}
\newtheorem{proposition}{Proposition}
\newtheorem{conjecture}{Conjecture}
\newtheorem{remark}{Remark}

\def \diag{\operatornamewithlimits{diag}}
\def \min{\operatornamewithlimits{min}}
\def \max{\operatornamewithlimits{max}}
\def \log{\operatorname{log}}
\def \max{\operatorname{max}}
\def \rank{\operatorname{rank}}
\def \out{\operatorname{out}}
\def \exp{\operatorname{exp}}
\def \arg{\operatorname{arg}}
\def \E{\operatorname{E}}
\def \tr{\operatorname{tr}}
\def \SNR{\operatorname{SNR}}
\def \dB{\operatorname{dB}}
\def \ln{\operatorname{ln}}

\def \be {\begin{eqnarray}}
\def \ee {\end{eqnarray}}
\def \ben {\begin{eqnarray*}}
\def \een {\end{eqnarray*}}

\begin{abstract}
Opportunistic interference alignment (OIA) asymptotically achieves
the optimal degrees-of-freedom (DoF) in interfering multiple-access
channels (IMACs) in a distributed fashion, as a certain user scaling
condition is satisfied. For the multiple-input multiple-output IMAC,
it was shown that the singular value decomposition (SVD)-based
beamforming at the users fundamentally reduces the user scaling
condition required to achieve any target DoF compared to that for
the single-input multiple-output IMAC. In this paper, we tackle two
practical challenges of the existing SVD-based OIA: 1) the need of
full feedforward of the selected users' beamforming weight vectors
and 2) a low rate achieved based on the exiting zero-forcing (ZF)
receiver. We first propose a codebook-based OIA, in which the weight
vectors are chosen from a pre-defined codebook with a finite size so
that information of the weight vectors can be sent to the belonging
BS with limited feedforward. We derive the codebook size required to
achieve the same user scaling condition as the SVD-based OIA case
for both Grassmannian and random codebooks. Surprisingly, it is
shown that the derived codebook size is the same for the two
considered codebook approaches. Second, we take into account an
enhanced receiver at the base stations (BSs) in pursuit of improving
the achievable rate based on the ZF receiver. Assuming no
collaboration between the BSs, the interfering links between a BS
and the selected users in neighboring cells are difficult to be
acquired at the belonging BS. We propose the use of a simple minimum
Euclidean distance receiver operating with no information of the
interfering links. With the help of the OIA, we show that this new
receiver asymptotically achieves the channel capacity as the number
of users increases.
\end{abstract}


\newpage


\section{Introduction}
Interference alignment (IA) \cite{V_Cadambe08_TIT,
M_MaddahAli08_TIT} is the key ingredient to achieve the optimal
degrees-of-freedom (DoF) for a variety of interference channel
models. The conventional IA framework, however, has several
well-known practical challenges: global channel state information
(CSI) and arbitrarily large frequency/time-domain dimension
extension. Recently, the concept of opportunistic interference
alignment (OIA) was introduced in \cite{B_Jung11_CL,B_Jung11_TC},
for the $K$-cell single-input multiple-output (SIMO) interfering
multiple-access channel (IMAC), where there are one $M$-antenna base
station and $N$ users in each cell. In the OIA scheme for the SIMO
IMAC, $S$ ($S\le M$) users amongst the $N$ users are
opportunistically selected in each cell in the sense that inter-cell
interference is aligned at a pre-defined interference space. Even if
several studies have independently addressed one or a few of the
practical problems~(see~\cite{C_Suh08_Allerton,
K_Comadam08_GLOBECOM}), the OIA scheme simultaneously resolves the
aforementioned issues. Specifically, the OIA scheme operates with i)
local CSI acquired via pilot signaling, ii) no dimension extension
in the time/frequency domain, iii) no iterative optimization of
precoders, and iv) no coordination between the users or the BSs. It
has been shown that there exists a trade-off between the the
achievable DoF and the number of users, which can be characterized
by a \textit{user scaling condition} \cite{B_Jung11_TC,
H_Yang13_TWC, H_Yang12_ISIT}. Similarly, the analysis of the scaling
condition of some system parameters required to achieve a target
performance have been widely studied to provide a remarkable insight
into the convergence rate to the target performance, e.g., the user
scaling condition to achieve target DoF for the IMAC
\cite{B_Jung11_CL,B_Jung11_TC, H_Yang13_TWC, H_Yang12_ISIT}, the
scaling condition of the number of feedback bits to achieve the
optimal DoF for multiple-input multiple-output (MIMO) interference
channels \cite{J_Thukral09_ISIT, R_Krishnamachari10_ISIT}, and the
codebook size scaling condition to achieve the target achievable
rate for limited feedback MIMO systems
\cite{B_Mondal06_TSP,T_Yoo07_JSAC,N_Jindal06_TIT}. For the SIMO
IMAC, the OIA scheme asymptotically achieves $KS$ DoF, for $0<S\le
M$, if the number of per-cell users, $N$, scales faster than
$\textrm{SNR}^{(K-1)S}$ \cite{B_Jung11_TC}, where SNR denotes the
received signal-to-noise ratio (SNR). Note that the optimal DoF is
achieved when $S=M$.

For the MIMO IMAC, where each user has $L$ antennas, the user
scaling condition to achieve $KS$ DoF can be greatly reduced to
$\textrm{SNR}^{(K-1)S-L+1}$ with the use of singular value
decomposition (SVD)-based beamforming at each user, by further
minimizing the generating interference level~\cite{H_Yang13_TWC}.
However, to implement the SVD-based OIA with local CSI and no
coordination between the users or the BSs, each beamforming weight
vector is computed at each user, and then information of the
selected users' weight vectors should be sent to the corresponding
BS for the coherent detection. In addition, although the OIA based
on the zero-forcing (ZF) receiver at the BSs is sufficient to
achieve the optimal DoF, its achievable rate is in general far below
the channel capacity, and the gap increases as the dimension of
channel matrices grows.
 In this paper, we would like to answer the aforementioned two practical issues of the SVD-based OIA.

In recent cellular systems such as the 3GPP Long Term Evolution
\cite{TS36.213}, each selected user should transmit an uplink pilot
(known as Sounding Reference Signal in 3GPP systems) so that the
corresponding BS estimates the uplink channel matrix, which is
widely used for channel quality estimation, downlink signal design
assuming the channel reciprocity in time division duplexing (TDD)
systems, etc. The effective channel matrices rotated by the weight
vectors should also be known so that the BSs perform coherent
detection---the matrices can be acquired by the BSs through either
of the following two methods: i) additional dedicated time/frequency
pilot (known as Demodulation Reference Signal in 3GPP systems
\cite{TS36.213}), where the pilots are rotated by weight vectors
\cite{L_Choi04_TWC, Z_Pan04_TWC} and ii) limited feedforward of the
indices of the weight vectors (as included in Downlink Control
Information Format 4 \cite{TS36.212}). For the first method,
however, the system capacity can be degraded as the number of
selected users increases due to the increased pilot overhead
\cite{C_Chae08_TSP, I_Hwang10_USPatent}. For a reliable
transmission, the length of pilot signaling also needs to be
sufficiently long~\cite{J_Jose11_TVT,D_Samardzija06_CISS}.
Furthermore, in cellular networks, long training sequences or
disjoint pilot resources for all users in each cell are required to
avoid the pilot contamination coming from the inter-cell
interference \cite{J_Jose11_TWC}. For these reasons, practical
communication systems such as the 3GPP standard allow highly limited
resources for uplink pilot. On the other hand, the second method
using the limited feedforward is preferable especially for the MIMO
IMAC in the sense that feedforward information can be flexibly
multiplexed with uplink data requiring no additional pilot resource.
Several studies \cite{C_Chae08_TSP,
I_Hwang10_USPatent,L_Equigua11_SPL} have addressed the same issues
on the feedforward of the weight vectors for multiuser MIMO systems,
and have proposed the design of codebook-based precoding matrices.

In the first part of this paper, we introduce a codebook-based OIA
scheme, where weight vectors are chosen from a pre-defined codebook
with a finite size such that information of the weight vectors of
selected users is sent to the corresponding BS via limited
feedforward signaling. Two widely-used codebooks, the Grassmannian
and random codebooks, are used. Surprisingly, although the
granularity of the Grassmannian codebook is higher than that of the
random codebook for a given codebook size, our result indicates that
for both codebook approaches, the codebook size, required to achieve
the same user scaling condition as the SVD-based OIA case,
coincides. It is also shown that the required codebook size in bits
increases linearly with the number of transmit antennas and
logarithmically with the received SNR, i.e., the required codebook
size scales as $L\log_2\textrm{SNR}$.

In the second part, we propose a receiver design at the BSs in
pursuit of improving the achievable rate based on the ZF receiver.
The design is challenging in the sense that local CSI and no
coordination between any BSs are assumed, thus resulting in no
available information of the interfering links at each BS. Thus, the
maximum likelihood (ML) decoding is not possible at each BS since
the covariance matrix of the effective noise cannot be estimated due
to no information of the interfering links. We propose the use of a
simple minimum Euclidean distance receiver, where the ML
cost-function is used by assuming the identity noise covariance
matrix, which does not require information of the interfering links.
We show that this receiver asymptotically achieves the channel
capacity as the number of users increases.

Simulation results are provided to justify the derived user and
codebook size scaling conditions and to evaluate the performance of
the minimum Euclidean distance receiver. A practical scenario, e.g.,
low SNR, small codebook size, and a small number of users, is taken
into account to show the robustness of our scheme.

The remainder of this paper is organized as follows. Section
\ref{SEC:system} describes the system and channel models and Section
\ref{SEC:CB_OIA} presents the proposed codebook-based OIA scheme.
Section \ref{sec:achievability} derives the user and codebook size
scaling conditions of the proposed OIA scheme along with two
different codebooks and Section \ref{sec:acvanced_Rx} derives the
asymptotic performance of the minimum Euclidean distance receiver.
Section \ref{SEC:Sim} performs the numerical evaluation. Section
\ref{SEC:Conc} summarizes the paper with some concluding remarks.

\textit{Notations:} $\mathbb{C}$ indicates the field of complex
numbers. $(\cdot)^{\textrm{T}}$ and $(\cdot)^{\textrm{H}}$ denote
the transpose and the conjugate transpose, respectively.
$\mathbf{I}_n$ is the $(n\times n)$-dimensional identity matrix.

\section{System and Channel Models}\label{SEC:system}
Consider the TDD MIMO IMAC with $K$ cells, each of which consists of
a BS with $M$ antennas and $N$ users, each having $L$ antennas, as
depicted in Fig. \ref{fig:sys}. It is assumed that each selected
user transmits a single spatial stream. In each cell, $S$ ($S\le M$)
users are selected for uplink transmission.
Let $\mathbf{H}_{k}^{[i,j]}\in \mathbb{C}^{M \times L}$ denote the
channel matrix from user $j$ in the $i$th cell to BS $k$. A
frequency-flat fading and the reciprocity between uplink and
downlink channels are assumed. Each element of
$\mathbf{H}_{k}^{[i,j]}$ is assumed to be an identical and
independent complex Gaussian random variable with zero mean and
variance $1/L$. User $j$ in the $i$th cell estimates the uplink
channel of its own link, $\mathbf{H}^{[i,j]}_k$ ($k=1, \ldots, K$),
via downlink pilot signaling transmitted from the BSs; that is,
local CSI is utilized as in \cite{K_Comadam08_GLOBECOM}.
Without loss of generality, the indices of the selected users in
each cell are assumed to be $(1, \ldots, S)$ for notational
simplicity. Then, the received signal at BS $i$ is expressed as:
\begin{align} \label{eq:y}
\mathbf{y}_i &= \sum_{j=1}^{S}\mathbf{H}_{i}^{[i,j]}
\mathbf{w}^{[i,j]}x^{[i,j]}  +\underbrace{\sum_{k=1, k\neq i}^{K}
\sum_{m=1}^{S}
\mathbf{H}_{i}^{[k,m]}\mathbf{w}^{[k,m]}x^{[k,m]}}_{\textrm{inter-cell
interference}} + \mathbf{z}_i,
\end{align}
where $\mathbf{w}^{[i,j]}\in \mathbb{C}^{L \times 1}$ and
$x^{[i,j]}$ are the weight vector and transmit symbol with unit
average power at user $j$ in the $i$th cell, respectively, and
$\mathbf{z}_i \in \mathbb{C}^{M \times 1}$ denotes the additive
white Gaussian noise at BS $i$, with zero mean and the covariance
$N_0\mathbf{I}_M$.

\begin{figure}
\begin{center}
  \includegraphics[width=0.46\textwidth]{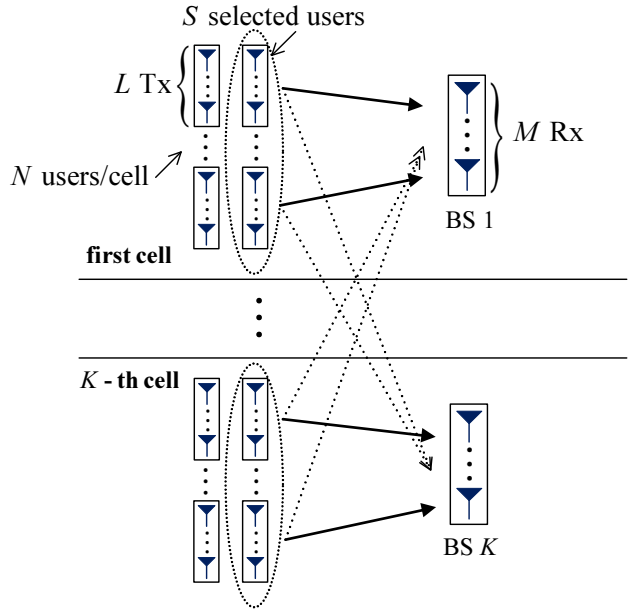}\\
  \caption{The MIMO IMAC model.}\label{fig:sys}
  \end{center}
\end{figure}

\section{Proposed Codebook-Based OIA: Overall Procedure} \label{SEC:CB_OIA}
The proposed scheme essentially follows the same procedure as that
of the SVD-based MIMO OIA \cite{H_Yang13_TWC, H_Yang12_ISIT} except
for the weight vector design step. For the completeness of our
achievability results, we briefly describe the overall procedure for
all the steps.

\subsection{Offline Procedure - Reference Basis Broadcasting}
The orthogonal reference basis matrix at BS $k$, to which the
received interference vectors are aligned, is denoted by
$\mathbf{Q}_k = \left[ \mathbf{q}_{k,1}, \ldots,
\mathbf{q}_{k,M-S}\right]\in \mathbb{C}^{M \times (M-S)}$. Here, BS
$k$ in the $k$th cell ($k\in  \{1, \ldots, K\}$) independently and
randomly generates $\mathbf{q}_{k,m}\in \mathbb{C}^{M \times 1}$ ($m
= 1, \ldots, M-S$) from the $M$-dimensional sphere. BS $k$ also
finds the null space of $\mathbf{Q}_k$, defined by $\mathbf{U}_k =
\left[ \mathbf{u}_{k,1}, \ldots, \mathbf{u}_{k,S} \right] \triangleq
\textrm{null}(\mathbf{Q}_k)$, where $\mathbf{u}_{k,i}\in
\mathbb{C}^{M \times 1}$ is orthonormal, and then broadcasts it to
all users.
Note that this process is required only once prior to data
transmission and does not need to change with respect to channel
instances.


%
%
%

\subsection{Step 1 - Weight Vector Design}
Let us denote the codebook set consisting of $N_f$ elements as
$\mathcal{C}_f \triangleq \left\{\mathbf{c}_1, \ldots,
\mathbf{c}_{N_f} \right\}$, where $\mathbf{c}_1, \ldots,
\mathbf{c}_{N_f}\in \mathbb{C}^{L \times 1}$ are chosen from the
$L$-dimensional unit sphere. Then, the number of bits to represent
$\mathcal{C}_f$ is denoted by $n_f = \lceil \log_2 N_f \rceil$. Let
$\mathbf{w}^{[i,j]}$ denote the weight vector at user $j$ in the
$i$th cell. Each user attempts to minimize the leakage of
interference (LIF) $\eta^{[i,j]}$ defined
by~\cite{B_Jung11_TC,H_Yang13_TWC}
\begin{align} \label{eq:LIF_VQ_def}
\eta^{[i,j]} = \sum_{k=1, k\neq i}^{K}
\left\|\mathbf{U}_k^{\textrm{H}}\mathbf{H}_{k}^{[i,j]}
\mathbf{w}^{[i,j]} \right\|^2 = \left\|\mathbf{G}^{[i,j]}
\mathbf{w}^{[i,j]} \right\|^2,
\end{align}
where $\mathbf{G}^{[i,j]}$ is the stacked interference channel
matrix given by
\begin{align} \label{eq:G_def}
\mathbf{G}^{[i,j]} &\triangleq \Big[
\left({\mathbf{U}_1}^{\textrm{H}}\mathbf{H}_{1}^{[i,j]}\right)^{\textrm{T}},
\ldots,
\left({\mathbf{U}_{i-1}}^{\textrm{H}}\mathbf{H}_{i-1}^{[i,j]}\right)^{\textrm{T}},
\left({\mathbf{U}_{i+1}}^{\textrm{H}}\mathbf{H}_{i+1}^{[i,j]}\right)^{\textrm{T}},
\ldots,
\left({\mathbf{U}_K}^{\textrm{H}}\mathbf{H}_{K}^{[i,j]}\right)^{\textrm{T}}
\Big]^{\textrm{T}}.
\end{align}
Let us denote the SVD of $\mathbf{G}^{[i,j]}$ by
$\mathbf{G}^{[i,j]}=\mathbf{E}^{[i,j]}\boldsymbol{\Sigma}^{[i,j]}{\mathbf{V}^{[i,j]}}^{\textrm{H}}$,
where $\mathbf{E}^{[i,j]}\in \mathbb{C}^{(K-1)S\times L}$ and
$\mathbf{V}^{[i,j]}\in \mathbb{C}^{L\times L}$ are left- and
right-singular vectors of $\mathbf{G}^{[i,j]}$, respectively,
consisting of $L$ orthonormal columns, and
$\boldsymbol{\Sigma}^{[i,j]} = \textrm{diag}\left(
\sigma^{[i,j]}_{1}, \ldots, \sigma^{[i,j]}_{L}\right)$. Here,
$\sigma^{[i,j]}_{m}$ denotes the $m$th singular value of
$\mathbf{G}^{[i,j]}$, where $\sigma^{[i,j]}_{1}\ge \cdots
\ge\sigma^{[i,j]}_{L}$. Then, $\mathbf{w}^{[i,j]}$ is obtained from
$\mathbf{w}^{[i,j]} = \arg \max_{1\le n \le N_f} \left|
\left(\mathbf{v}_{L}^{[i,j]}\right)^H \mathbf{c}_n\right|^2$.
Clearly, the weight vector minimizing $\eta^{[i,j]}$ is the $L$th
column of $\mathbf{V}^{[i,j]}$, denoted by $\mathbf{v}_{L}^{[i,j]}$,
and this precoding is subject to the SVD-based OIA.

\subsection{Step 2 - User Selection}
Each user reports its LIF metric in (\ref{eq:LIF_VQ_def}) to the
corresponding BS, and then each BS selects $S$ users, having the LIF
metrics up to the $S$th smallest one, amongst $N$ users in the cell.
Subsequently, each selected user forwards the index of
$\mathbf{w}^{[i,j]}$ in the codebook to the belonging BS.

\subsection{Step 3 - Uplink Transmission and Detection}
If all the selected users transmit the uplink signals
simultaneously, then the received signal at BS $i$ is given by
(\ref{eq:y}). As in the SVD-based OIA~\cite{H_Yang13_TWC}, the
linear ZF detection is sufficient to achieve the maximum DoF. The
decision statistics $\mathbf{r}_i$ at BS $i$ is obtained from
\begin{equation}\label{eq:r}
\mathbf{r}_i = \left[r_{i,1}, \ldots, r_{i,S}
\right]^{\textrm{T}}\triangleq
{\mathbf{F}_i}^{\textrm{H}}\mathbf{U}_i^{\textrm{H}}\mathbf{y}_i,
\end{equation}
where $\mathbf{F}_{i}\in \mathbb{C}^{S \times S}$ is the ZF
equalizer defined by
\begin{align}
\mathbf{F}_i &= \left[\mathbf{f}_{i,1}, \ldots, \mathbf{f}_{i,S}
\right]  \nonumber\\ & \triangleq
\left(\left[{\mathbf{U}_i}^{\textrm{H}}\mathbf{H}_{i}^{[i,1]}
\mathbf{w}^{[i,1]}, \ldots,
{\mathbf{U}_i}^{\textrm{H}}\mathbf{H}_{i}^{[i,S]}
\mathbf{w}^{[i,S]}\right]^{-1}\right)^{\textrm{H}}. \nonumber
\end{align}
Note that multiplying ${\mathbf{U}_i}^{\textrm{H}}$ to
$\mathbf{y}_i$ cancels interference aligned at $\mathbf{Q}_i$. The
achievable rate $R^{[i,j]}$ is then given by
\begin{align} \label{eq:rate_general}
R^{[i,j]} &\!=\! \log_2\left( 1+ \textrm{SINR}^{[i,j]}\right)
\!=\!\log_2 \left( 1+
\frac{\textrm{SNR}}{\left\|\mathbf{f}_{i,j}\right\|^2\! +\!
\tilde{I}^{[i,j]}  } \right),
\end{align}
where $\textrm{SNR} = 1/N_0$ and $\tilde{I}^{[i,j]} \triangleq
\sum_{k=1, k\neq i}^{K} \sum_{m=1}^{S} \left|
{\mathbf{f}_{i,j}}^{\textrm{H}}
{\mathbf{U}_i}^{\textrm{H}}\mathbf{H}_{i}^{[k,m]}\mathbf{w}^{[k,m]}
\right|^2 \textrm{SNR}$.

\begin{figure}
\begin{center}
  \includegraphics[width=0.73\textwidth]{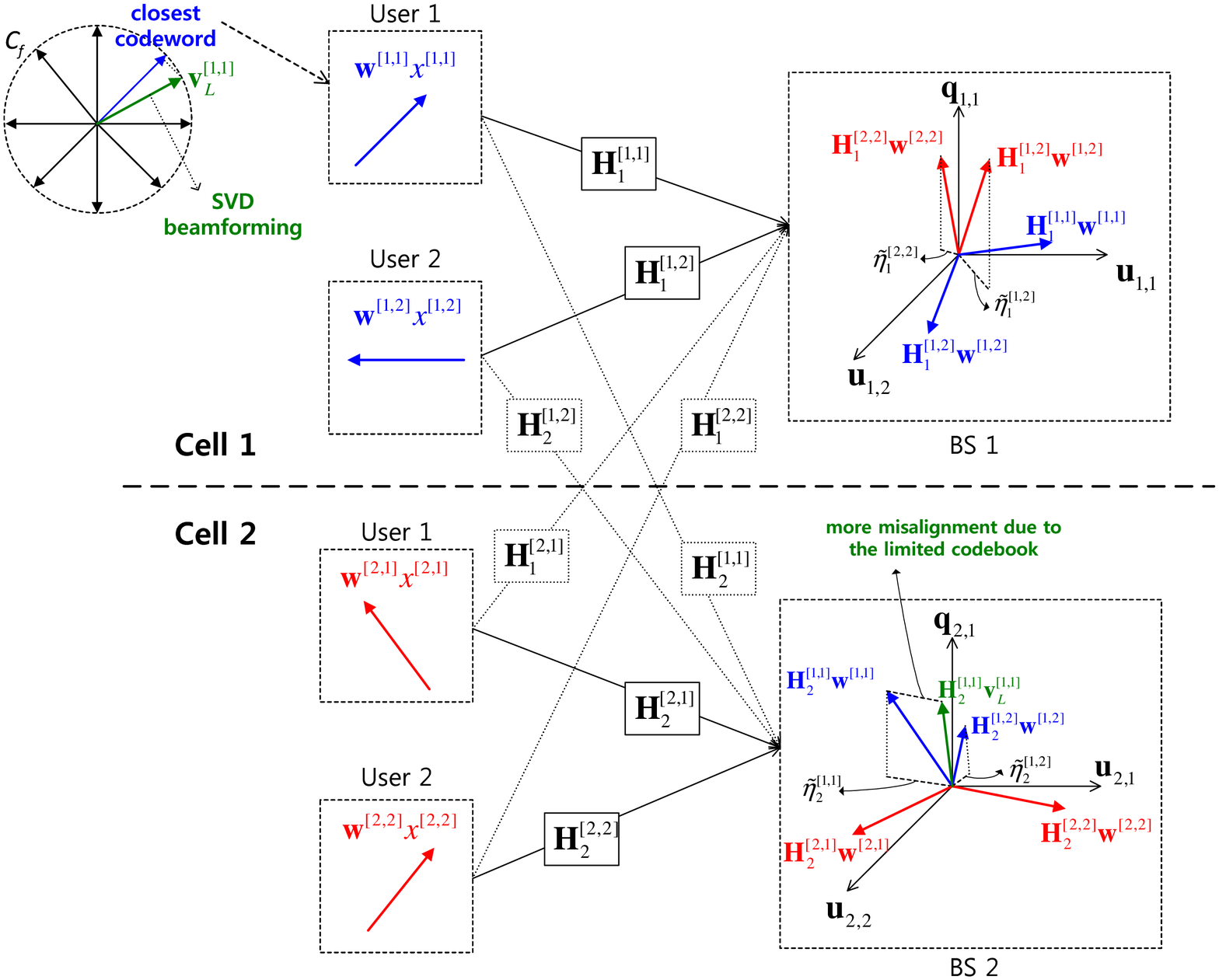}\\
  \caption{The codebook-based OIA where $K=2$, $M=3$, $S=2$.}\label{fig:OIA}
  \end{center}
\end{figure}

Figure \ref{fig:OIA} illustrates the principle of the proposed
signaling for $K=2$, $M=3$, and $S=2$. If interference
$\mathbf{H}_{i}^{[k,m]}\mathbf{w}^{[k,m]}$ in (\ref{eq:y}) is
perfectly aligned to the interference basis $\mathbf{Q}_i$, i.e.,
$\mathbf{H}_{i}^{[k,m]}\mathbf{w}^{[k,m]} \in
\textrm{span}(\mathbf{Q}_i)$, then interference in $\mathbf{r}_i$ of
(\ref{eq:r}) vanishes because
$\mathbf{U}_i^{\textrm{H}}\mathbf{H}_{i}^{[k,m]}\mathbf{w}^{[k,m]} =
\mathbf{0}$. As illustrated in Fig. \ref{fig:OIA}, the value
$\|\mathbf{U}_i^{\textrm{H}}\mathbf{H}_{i}^{[k,m]}
\mathbf{w}^{[k,m]} \|^2$ represents the amount of the signal
transmitted from user $m$ in the $k$th cell to BS $i$ that is not
aligned to the interference reference basis $\mathbf{Q}_i$. This
misalignment becomes higher than the SVD-based OIA case, due to the
finite codebook size.



\section{Achievability Results} \label{sec:achievability}
It was shown in~\cite{H_Yang13_TWC} that using the SVD-based OIA
scheme leads to a comparatively less number of users required to
achieve the maximum DoF in the MIMO IMAC model. In this section, we
derive the number of feedforward bits required to achieve the same
achievability as the SVD-based OIA case in terms of the DoF and user
scaling condition when two different types of codebook-based OIA
schemes, i.e., the Grassmannian and random codebook-based OIAs, are
used.

In our analysis, we use the total DoF defined
as~\cite{V_Cadambe08_TIT}
\begin{equation}
\textrm{DoF} = \lim_{\textrm{SNR} \rightarrow \infty}
\frac{\sum_{i=1}^{K}\sum_{j=1}^{S}R^{[i,j]}}{\log_2 \textrm{SNR}},
\nonumber
\end{equation}
where $R^{[i,j]}$ is the achievable rate for user $j$ in the $i$th
cell and $\textrm{SNR} = 1/N_0$.

\subsection{Grassmannian Codebook-Based OIA} \label{subsec:Grassmaniaan}
We start with the following three lemmas which shall be used to
establish our main theorem.

\begin{lemma} \label{lemma:eta_UB}
For any given codebook, the LIF metric $\eta^{[i,j]}$ in
(\ref{eq:LIF_VQ_def}) is upper-bounded by
\begin{equation}\label{eq:eta_lemma1}
\eta^{[i,j]} \le {\sigma^{[i,j]}_L}^2+   {d^{[i,j]}}^2
{\sigma_1^{[i,j]}}^2,
\end{equation}
where $d^{[i,j]}$ is the residual distance defined by
\begin{equation}\label{eq:d_ij}
d^{[i,j]} = \sqrt{1-\left|{\mathbf{w}^{[i,j]}}^{\textrm{H}}
\mathbf{v}_L^{[i,j]} \right|^2}
\end{equation}
and $\sigma^{[i,j]}_m$ is the $m$th singular value of the stacked
interference channel matrix in (\ref{eq:G_def}).
\end{lemma}

\begin{proof}
Since $\mathbf{v}_L^{[i,j]}$ is isotropically distributed over the
$L$-dimensional sphere with identically and isotropically
distributed (i.i.d.)  complex Gaussian channel matrices
\cite{D_Love03_TIT}, the weight vector $\mathbf{w}^{[i,j]}$ chosen
from a codebook can be written by $\mathbf{w}^{[i,j]} =
\sqrt{1-{d^{[i,j]}}^2}\mathbf{v}_L^{[i,j]}+d^{[i,j]}
\mathbf{t}^{[i,j]}$~\cite{W_Dai08_TIT, N_Jindal06_TIT}, where $0\le
{d^{[i,j]}}^2\le1$ accounts for the quantization error and
$\mathbf{t}^{[i,j]}$ is a unit-norm vector i.i.d. over
$\textrm{null}\left(\mathbf{v}_L^{[i,j]}\right)$.\pagebreak[0] Then,
$\eta^{[i,j]}$ in (\ref{eq:LIF_VQ_def}) is bounded by
\begin{align}
\eta^{[i,j]} &= \left\|\sqrt{1-{d^{[i,j]}}^2}\mathbf{G}^{[i,j]}\mathbf{v}^{[i,j]}_L + {d^{[i,j]}} \mathbf{G}^{[i,j]}\mathbf{t}^{[i,j]} \right\|^2 \nonumber\\
 &  \le (1-{d^{[i,j]}}^2){\sigma^{[i,j]}_L}^2 + {d^{[i,j]}}^2\left\| \mathbf{G}^{[i,j]}\mathbf{t}^{[i,j]} \right\|^2 \nonumber\\
\label{eq:eta_UB4} & \le {\sigma^{[i,j]}_L}^2+   {d^{[i,j]}}^2
{\sigma_1^{[i,j]}}^2,
\end{align}
where (\ref{eq:eta_UB4}) follows from $\left\|
\mathbf{G}^{[i,j]}\mathbf{t}^{[i,j]} \right\|^2  \le
{\sigma_1^{[i,j]}}^2$ for any unit-norm vector $\mathbf{t}^{[i,j]}$,
which proves the lemma.
\end{proof}

Now, we further bound the LIF metric for the Grassmannian codebook
as follows.

\begin{lemma}\label{lemma:eta_UB_Grass}
By using the Grassmannian codebook, $\eta^{[i,j]}$ is further
bounded by $\eta^{[i,j]} \le {\sigma^{[i,j]}_L}^2+   \nu_f
{\sigma_1^{[i,j]}}^2$, where $\nu_f$ denotes the number of
feedforward bits given by
\begin{equation} \label{eq:dmin_bound2}
\nu_f =  \left( \frac{1}{N_f}\right)^{1/(L-1)}
\end{equation}
and $N_f$ is the number of elements in the codebook set.
\end{lemma}
\begin{proof}
The Grassmannian codebook $\mathcal{C}_f$ is the set of codewords
chosen by  the optimal sphere packing for the $L$-dimensional
sphere; namely, the chordal distance of any two codewords is all the
same, i.e., $\sqrt{1-\left|{\mathbf{c}_i}^{\textrm{H}} \mathbf{c}_j
\right|^2}=d$ for any $i\neq j$ and $d\ge0$. Based on this property,
the Rankin, Gilbert-Varshamov, and Hamming bounds on the distance of
the codebook give us \cite{A_Barg02_TIT, W_Dai08_TIT}
\begin{equation} \label{eq:dmin_bound}
{d^{[i,j]}}^2 \le  \min \left\{\frac{1}{2},
\frac{(L-1)N_f}{2L(N_f-1)}, \left(
\frac{1}{N_f}\right)^{1/(L-1)}\right\}.
\end{equation}
For large $N_f$, the third term of (\ref{eq:dmin_bound}) becomes
dominant, thus providing an arbitrarily tight bound. Inserting
(\ref{eq:dmin_bound}) to (\ref{eq:eta_lemma1}) proves the lemma.
\end{proof}

From Lemma \ref{lemma:eta_UB_Grass}, we also have the following
lemma.
\begin{lemma} \label{lemma:eta_UB2}
For the Grassmannian codebook, it follows that $\eta^{[i,j]} \le
\eta^{[i,j]}_{\textrm{GC}}$ where
\begin{equation} \label{eq:eta_tilde}
\eta^{[i,j]}_{\textrm{GC}} = \left\{ \begin{array}{cc}
                         C_1\nu_f{\sigma_1^{[i,j]}}^2 & \textrm{ if }  {\sigma_L^{[i,j]}}^2 \le (1+\delta)\nu_f{\sigma_1^{[i,j]}}^2\\
                         C_2{\sigma_L^{[i,j]}}^2 & \textrm{ otherwise},
                       \end{array}
\right.
\end{equation}
for any constant $\delta \ge 0$ independent of SNR. Here, $C_1 =
(2+\delta)$ and $C_2 = (1+1/(1+\delta))$; thus, $1 \le C_2 \le 2$.
\end{lemma}

Now we are ready to show our first main theorem, which derives the
number of feedforward bits, required to achieve the same user
scaling condition as the SVD-based OIA case, for the proposed OIA
with Grassmannian codebook.

\begin{theorem}\label{theorem:CB}
The codebook-based OIA with the optimal Grassmannian codebook
$\mathcal{C}_f$ \cite{D_Love03_TIT} achieves the user scaling
condition of the SVD-based OIA if~\footnote{We use the following
notation: i) $f(x)=O(g(x))$ means that there exist constants $M$ and
$m$ such that $f(x)\le Mg(x)$ for all $x>m$. ii) $f(x)=o(g(x))$
means that $\underset{x\rightarrow\infty}\lim\frac{f(x)}{g(x)}=0$.
iii) $f(x)=\Omega(g(x))$ if $g(x)=O(f(x))$. iv) $f(x)=\omega(g(x))$
if $g(x)=o(f(x))$.}
\begin{align}\label{eq:N_f_cond}
n_f =\Omega\left(L \log_2 \textrm{SNR}\right) \textrm{ bits}.
\end{align} Moreover, under the condition (\ref{eq:N_f_cond}),
$KS$ DoF are achievable with high probability if $N =
\omega\left(\textrm{SNR}^{(K-1)S -L+1}\right)$.
\end{theorem}


\begin{proof}
Let us start from showing the following simple bound on
$\textrm{SINR}^{[i,j]}$ in (\ref{eq:rate_general}):
\begin{equation} \label{eq:SINR_LB}
\textrm{SINR}^{[i,j]} \ge
\frac{\textrm{SNR}/\left\|\mathbf{f}_{i,j}\right\|^2}{1 +
I^{[i,j]}},
\end{equation}
where $I^{[i,j]} \triangleq \sum_{k=1, k\neq i}^{K} \sum_{m=1}^{S}
\left\|
{\mathbf{U}_i}^{\textrm{H}}\mathbf{H}_{i}^{[k,m]}\mathbf{w}^{[k,m]}
\right\|^2 \textrm{SNR}$.
Suppose that $I^{[i,j]} \le \epsilon$ for some constant $\epsilon>0$
independent of the received SNR so that each user achieves 1 DoF. By
this principle, we obtain a lower bound on the achievable DoF for
the codebook-based OIA as $\textrm{DoF} \ge KS \cdot \mathcal{P}$,
where
\begin{align} \label{eq:P_def_CB}
\mathcal{P} \triangleq &\lim_{\textrm{SNR}\rightarrow
\infty}\textrm{Pr} \biggl\{I^{[i,j]}\le \epsilon, \forall
\textrm{user $j$ and BS $i$},  j\in
\mathcal{S}\triangleq\{1,\dots,S\}, i \in
\mathcal{K}\triangleq\{1,\dots,K\} \biggr\}.
\end{align}
From the fact that the sum of received interference at all the BSs
is equivalent to the sum of the LIF metrics~\cite{B_Jung11_TC,
H_Yang13_TWC}, i.e.,
 \begin{align} \label{eq:LIF_equivalence}
& \sum_{i=1}^{K}\sum_{k=1, k\neq i}^{K} \sum_{m=1}^{S}
\left\|{\mathbf{U}_{i}}^{\textrm{H}}\mathbf{H}_{i}^{[k,m]}\mathbf{w}^{[i,j]}\right\|^2
=\sum_{i=1}^{K}\sum_{j=1}^{S} \eta^{[i,j]},
 \end{align}
we can bound $\mathcal{P}$ as
\begin{align}
\mathcal{P} &\ge \lim_{\textrm{SNR}\rightarrow \infty} \textrm{Pr} \bigg\{\sum_{i=1}^{K}\sum_{j=1}^{S}I^{[i,j]}\le \epsilon\bigg\}\nonumber \\
 \label{eq:P_CB_LB2} &\ge \lim_{\textrm{SNR}\rightarrow \infty} \textrm{Pr} \left\{\eta^{[i,j]} \le \frac{\epsilon\textrm{SNR}^{-1}}{KS^2}, \forall i\in \mathcal{K}, \forall j\in\mathcal{S}\right\}.
\end{align}
At this point, let us choose $N_f$ such that $N_f^{1/(L-1)} \ge
\textrm{SNR}^{(1+\gamma)}$ for $\gamma>0$, i.e.,
\begin{equation}\label{eq:n_f_choose}
n_f = \log_2 N_f \ge (1+\gamma)(L-1)\log_2 \textrm{SNR},
\end{equation}
resulting in (\ref{eq:N_f_cond}).
From (\ref{eq:dmin_bound2}) and (\ref{eq:n_f_choose}), $\nu_f$ is
bounded by
\begin{equation} \label{eq:N_f_inequality}
\nu_f = N_f^{-1/(L-1)} \le \textrm{SNR}^{-(1+\gamma)}.
\end{equation}
Now we consider the LIF-overestimating modification by using the
upper bound $\eta^{[i,j]}_{\textrm{GC}}$ in Lemma
\ref{lemma:eta_UB2}. From $ \eta^{[i,j]} \le
\eta^{[i,j]}_{\textrm{GC}}$ and (\ref{eq:P_CB_LB2}), we have
\begin{align}
\mathcal{P} &\ge \lim_{\textrm{SNR}\rightarrow \infty}\textrm{Pr} \left\{ \eta^{[i,j]}_{\textrm{GC}}\le \frac{\epsilon \textrm{SNR}^{-1}}{KS^2}, \forall i\in \mathcal{K}, \forall j \in \mathcal{S}\right\}.\label{eq:P_prime_LB0}\\
& \ge  \mathcal{P}_{\textrm{GC}} \triangleq \lim_{\textrm{SNR}\rightarrow \infty}\textrm{Pr} \bigg\{ \left[\eta^{[i,j]}_{\textrm{GC}}\le \frac{\epsilon \textrm{SNR}^{-1}}{KS^2}, \forall i\in \mathcal{K}, \forall j \in \mathcal{S}\right] \nonumber \\
\label{eq:P_prime_LB1}& \& \left[{\sigma_L^{[i,j]}}^2 \!\ge\!
(1+\delta)\nu_f {\sigma_1^{[i,j]}}^2, \forall i \in \mathcal{K},
\forall j\in \mathcal{N}\triangleq\{1,\dots,N\}\right]\bigg\}.
\end{align}
From the principle $\textrm{Pr}(\mathcal{A}\cap \mathcal{B}) =
\textrm{Pr}(\mathcal{B}) - \textrm{Pr}(\mathcal{A}^c\cap
\mathcal{B})$ for sets $\mathcal{A}$ and $\mathcal{B}$,
(\ref{eq:P_prime_LB1}) can be rewritten as
\begin{align}
\pagebreak[0]& \mathcal{P}_{\textrm{GC}} = \lim_{\textrm{SNR}\rightarrow \infty}  \underbrace{\textrm{Pr}\left\{ {\sigma_L^{[i,j]}}^2 \ge (1+\delta)\nu_f {\sigma_1^{[i,j]}}^2, \forall i\in\mathcal{K}, \forall j\in\mathcal{N} \right\}}_{\triangleq p_c} \label{eq:P_prime_LB1_00} \pagebreak[0] \\
& - \lim_{\textrm{SNR}\rightarrow \infty}\textrm{Pr} \bigg\{ \biggl[\textrm{there exist less than $S$ users per cell such that } \eta^{[i,j]}_{\textrm{GC}}\!\le\! \frac{\epsilon \textrm{SNR}^{-1}}{KS^2}\biggr] \nonumber\\ & ~~~~~\& \left[{\sigma_L^{[i,j]}}^2 \!\ge\! (1+\delta)\nu_f {\sigma_1^{[i,j]}}^2, \forall i \in \mathcal{K}, \forall j\in \mathcal{N}\right]\bigg\} \linebreak[0] \nonumber \\
& = \lim_{\textrm{SNR}\rightarrow \infty} p_c
-\lim_{\textrm{SNR}\rightarrow \infty}  \sum_{m=0}^{S-1} \left(
\begin{array}{c}
                                      N \\
                                      m
                                    \end{array} \right) \bigg[\textrm{Pr} \left\{ \eta^{[i,j]}_{\textrm{GC}} \le \frac{\epsilon \textrm{SNR}^{-1}}{KS^2} \hspace{3pt}\& \hspace{5pt}(1+\delta)\nu_f{\sigma_1^{[i,j]}}^2 \le {\sigma_L^{[i,j]}}^2 \right\} \bigg]^{m} \nonumber \displaybreak[0]\\
\label{eq:P_prime_LB1_2}& \times\bigg[\underbrace{\textrm{Pr} \left\{ (1+\delta)\nu_f{\sigma_1^{[i,j]}}^2 \le {\sigma_L^{[i,j]}}^2 ~ \& ~ \frac{\epsilon \textrm{SNR}^{-1}}{C_2 KS^2}\le {\sigma_L^{[i,j]}}^2 \right\}}_{\triangleq P_{o}}\bigg]^{N-m}\displaybreak[0]\\
\label{eq:P_prime_LB2}& \ge \lim_{\textrm{SNR}\rightarrow \infty}
p_c - \lim_{\textrm{SNR}\rightarrow \infty}  \sum_{m=0}^{S-1} N^m
{P_o}^{N-m},
\end{align}
where (\ref{eq:P_prime_LB1_2}) follows from the fact that the
statistics of each user is independent of each other, and
(\ref{eq:P_prime_LB2}) follows from $\textrm{Pr} \left\{
\eta^{[i,j]}_{\textrm{GC}} \le \frac{\epsilon
\textrm{SNR}^{-1}}{KS^2} \hspace{3pt}\&
\hspace{5pt}(1+\delta)\nu_f{\sigma_1^{[i,j]}}^2 \le
{\sigma_L^{[i,j]}}^2 \right\}  \le 1$  and $\bigg( \begin{array}{c}
                                      N \\
                                      i
                                    \end{array} \bigg) = \frac{N!}{i!(N-i)!}\le N^i$.

For the rest of the proof, we show that (\ref{eq:P_prime_LB2}) tends
to one under certain conditions. In Appendix A, we first show that
for given $\gamma>0$ and $\delta>0$, it follows that
\begin{eqnarray}
\lim_{\textrm{SNR}\rightarrow \infty} p_c =1, \textrm{ if $N = O
\left(\textrm{SNR}^{(1+\beta)((K-1)S-L+1)}\right)$}
\label{eq:p_c_asympt} \textrm{~~where $\beta< \gamma$.} \nonumber
\end{eqnarray}
Now we show that the second term of (\ref{eq:P_prime_LB2}) tends to
zero as the SNR increases. From the fact that $\textrm{Pr}\{[A \le
B] \& [C \le B] \} = \textrm{Pr}\{[A \le B] | A\ge C
\}\textrm{Pr}\{A\ge C\} +\textrm{Pr}\{[C \le B] | A< C
\}\textrm{Pr}\{A< C\}$ for random variables $A$, $B$, and $C$, the
probability $P_o$ can be written as
\begin{align}
P_o =& \textrm{Pr} \biggl\{ (1+\delta)\nu_f{\sigma_1^{[i,j]}}^2 \le {\sigma_L^{[i,j]}}^2  \Big| (1+\delta)\nu_f{\sigma_1^{[i,j]}}^2 \ge \frac{\epsilon \textrm{SNR}^{-1}}{C_2 KS^2}\biggr\} \cdot p_o^{\prime} \nonumber \\
& +  \textrm{Pr} \left\{ \frac{\epsilon \textrm{SNR}^{-1}}{C_2 KS^2}
\le {\sigma_L^{[i,j]}}^2  \Big| (1+\delta)\nu_f{\sigma_1^{[i,j]}}^2
< \frac{\epsilon \textrm{SNR}^{-1}}{C_2 KS^2}\right\} \cdot
(1-p_o^{\prime}) \nonumber
\end{align}
where $p_o^{\prime} = \textrm{Pr} \left\{
(1+\delta)\nu_f{\sigma_1^{[i,j]}}^2 \ge \frac{\epsilon
\textrm{SNR}^{-1}}{C_2 KS^2}\right\}$. From
(\ref{eq:N_f_inequality}), for any given channel instance, we have
\begin{align}
\lim_{\textrm{SNR}\rightarrow \infty} p_o^{\prime} & \le \lim_{\textrm{SNR}\rightarrow \infty} \textrm{Pr}\left\{ (1+\delta)\textrm{SNR}^{-(1+\gamma)}{\sigma_1^{[i,j]}}^2\ge \frac{\epsilon \textrm{SNR}^{-1}}{C_2KS^2}\right\} \nonumber\\
& = \lim_{\textrm{SNR}\rightarrow \infty} \textrm{Pr}\left\{
{\sigma_1^{[i,j]}}^2\ge
\frac{\epsilon\textrm{SNR}^{\gamma}}{(1+\delta) C_2KS^2}\right\}=0
\nonumber,
\end{align}
which results in
\begin{align}
\lim_{\textrm{SNR}\rightarrow \infty} P_o &=
\lim_{\textrm{SNR}\rightarrow \infty} \textrm{Pr} \left\{
\frac{\epsilon\textrm{SNR}^{-1}}{C_2KS^2}\le {\sigma_L^{[i,j]}}^2
\right\}. \nonumber
\end{align}
From \cite[Theorem 4]{S_Jin08_TC}, we have
\begin{align} \label{eq:Po_last}
&\textrm{Pr} \left\{ \frac{\epsilon \textrm{SNR}^{-1}}{C_2KS^2}\le
{\sigma_L^{[i,j]}}^2 \right\} \nonumber\\&= 1 - \alpha\left(
\frac{\epsilon }{C_2KS^2}\right)^{\psi} \textrm{SNR}^{-\psi}+
o\left( \textrm{SNR}^{-\psi}\right),
\end{align}
where $\psi \triangleq (K-1)S-L+1$ and $\alpha>0$ is a constant
determined by $K$, $S$, and $L$. Applying (\ref{eq:Po_last}) to
(\ref{eq:P_prime_LB2}) yields
\begin{align}
\label{eq:P_prime_LB_final}&\mathcal{P}\ge \mathcal{P}_{\textrm{GC}}
\ge \lim_{\textrm{SNR}\rightarrow \infty}  p_c -
\lim_{\textrm{SNR}\rightarrow \infty} \sum_{i=0}^{S-1} N^m \Biggl[1
- \alpha\left( \frac{\epsilon }{C_2KS^2}\right)^{\psi}
\textrm{SNR}^{-\psi} \hspace{15pt} + o\left(
\textrm{SNR}^{-\psi}\right)\Biggr]^{N-m}.
\end{align}
For given $\gamma$ and $0< \beta <\gamma$, let us choose
$N=O\left(\textrm{SNR}^{(1+\beta)((K-1)S-L+1)}\right)$. Then, from
(\ref{eq:p_c_asympt}), it follows that $
\lim_{\textrm{SNR}\rightarrow \infty} p_c =1$. On the other hand,
the second term of (\ref{eq:P_prime_LB_final}) tends to zero because
$N^m$ increases polynomially with SNR for given $m$ while $\left[1 -
\alpha\left( \frac{\epsilon }{C_2KS^2}\right)^{\psi}
\textrm{SNR}^{-\psi}+ o\left(
\textrm{SNR}^{-\psi}\right)\right]^{N-m}$ decreases exponentially
with SNR. Thus, $\mathcal{P}$ tends to one, which means that $KS$
DoF are achievable.

As assumed earlier, note that our analysis holds for $\beta<\gamma$.
However, it is obvious that assuming either the condition $N=
O\left(\textrm{SNR}^{(1+\beta)((K-1)S-L+1)}\right)$ for any
$\beta\ge \gamma$ or $N=
\omega\left(\textrm{SNR}^{(K-1)S-L+1}\right)$ leads to the same or
higher DoF compared to the case for $0\le \beta <\gamma$, due to the
fact that increasing $N$ for given $n_f$ and SNR values yields a
reduced LIF and thus an increased achievable rate for all the
selected users. Since the maximum achievable DoF are upper-bounded
by $KS$ for given $S$, the last argument indicates that $KS$ DoF are
achievable if $n_f=\Omega(L\log_2\textrm{SNR})$ and
$N=\omega\left(\textrm{SNR}^{(K-1)S-L+1}\right)$, which completes
the proof.
\end{proof}

Theorem \ref{theorem:CB} indicates that $n_f$ should scale with SNR
so as to achieve the target DoF under the same user scaling
condition as the SVD-based OIA case, and that from
(\ref{eq:n_f_choose}), no more feedforward bits than
$(1+\gamma)(L-1) \log_2 \textrm{SNR}$ are indeed required. The
derived $n_f$ scaling condition is proportional to $L$ and
$\log_2\textrm{SNR}$, which is consistent with the previous results
on the number of feedback bits required to avoid performance loss
due to the finite codebook size in a variety of limited feedback
systems \cite{D_Love03_TIT,N_Jindal06_TIT,J_Thukral09_ISIT}.

\subsection{Random Codebook-Based OIA}
For a random codebook scenario, each element $\mathbf{c}_n$ of
$\mathcal{C}_f$ ($n\in \{1,\ldots, N_f\}$) is chosen independently
and isotropically from the $L$-dimensional sphere. The following
second main theorem shows that the same user scaling condition as
the Grassmannian codebook-based OIA case is obtained even with the
random codebook-based OIA.

\begin{theorem} \label{theorem:RVQ}
The codebook-based OIA with a random codebook achieves $KS$ DoF with
high probability if $N = \omega\left( \textrm{SNR}^{(K-1)S-L+1}
\right)$ and $n_f =\Omega\left(L\log_2 \textrm{SNR}\right)$ bits.
\end{theorem}

\begin{proof}
Since equations (\ref{eq:SINR_LB})--(\ref{eq:P_CB_LB2}) also hold
for the random codebook approach, we only show that $\mathcal{P}$ in
(\ref{eq:P_def_CB}) tends to one under two conditions $N =
\omega\left( \textrm{SNR}^{(K-1)S-L+1} \right)$ and $n_f
=\Omega\left(L\log_2 \textrm{SNR}\right)$.
Unlike the Grassmannian codebook, the residual distance $d^{[i,j]}$
in (\ref{eq:d_ij}) is now a random variable and thus is unbounded.
Note that the cumulative density function (CDF) of the squared
chordal distance between any two independent unit random vectors
chosen isotropically from the $L$-dimensional sphere is given by
$\beta(L-1,1)$, where $\beta(x,y) =
\int_{0}^{1}t^{x-1}(1-t)^{y-1}dt$ is the beta function
\cite{N_Jindal06_TIT}. Since ${d^{[i,j]}}^2$ for the random codebook
is the minimum of $N_f$ independent random variables with
distribution  $\beta(L-1,1)$, the CDF of ${d^{[i,j]}}^2$ is given by
\begin{equation} \label{eq:F_d_def}
 \textrm{Pr}\left\{{d^{[i,j]}}^2\le z\right\} =  1-\left(1-z^{L-1}\right)^{N_f}.
\end{equation}
%
Now, let us again consider the following modification for given
channel instance:
\begin{itemize}
 \item[i)] if ${d^{[i,j]}}^2 \le \textrm{SNR}^{-(1+\gamma)}$  for all $i$ and $j$, then the same LIF-overestimating modification as the Grassmannian codebook case is used, where the LIF values are replaced with their upper bounds.
 Specifically, from Lemma \ref{lemma:eta_UB}, we shall use the following upper bound on $\eta^{[i,j]}$:
\begin{equation} \label{eq:eta_tilde2}
\eta^{[i,j]}_{\textrm{RC}}\!\! = \!\!\left\{ \begin{array}{cc}
                         \!\!\!C_3{d^{[i,j]}}^2{\sigma_1^{[i,j]}}^2 & \textrm{if~}  {\sigma_L^{[i,j]}}^2 \!\!\le\!\! (1+\delta^{\prime}){d^{[i,j]}}^2{\sigma_1^{[i,j]}}^2\\
                         \!\!\!C_4{\sigma_L^{[i,j]}}^2 & \textrm{otherwise},
                       \end{array}
\right.
\end{equation}
for any constant $\delta^{\prime} \ge 0$ independent of SNR, where
$C_3 = (2+\delta^{\prime})$ and $C_4 = (1+1/(1+\delta^{\prime}))$;
thus, $1 \le C_4 \le 2$,
\item[ii)] otherwise, i.e., if ${d^{[i,j]}}^2 > \textrm{SNR}^{-(1+\gamma)}$ for any $i$ or $j$, then we drop the case by assuming 0 DoF for this case.
\end{itemize}
Let us define the event $\mathcal{D}$ as
\begin{align}
\mathcal{D}& = \bigl\{ {d^{[i,j]}}^2 \le \textrm{SNR}^{-(1+\gamma)},
\forall i\in \mathcal{K}=\{1,\dots,K\}, \forall
j\in\mathcal{N}=\{1,\dots,N\}\bigr\}. \nonumber
 \end{align}
 From $\eta^{[i,j]} \le \eta^{[i,j]}_{\textrm{RC}}$, we have
\begin{align}
\mathcal{P} &\ge \lim_{\textrm{SNR}\rightarrow \infty} \textrm{Pr} \biggl\{\eta^{[i,j]} \le \frac{\epsilon\textrm{SNR}^{-1}}{KS^2}, \forall i\in \mathcal{K},\forall j\in\mathcal{S}=\{1,\dots,S\}\biggr\} \nonumber\\
&\ge \mathcal{P}_{\textrm{RC}}
\triangleq\lim_{\textrm{SNR}\rightarrow \infty} \textrm{Pr} \left\{
\mathcal{D}\right\} \cdot\textrm{Pr} \left\{
\eta^{[i,j]}_{\textrm{RC}}\le \frac{\epsilon
\textrm{SNR}^{-1}}{KS^2}, \forall i\in \mathcal{K}, \forall j \in
\mathcal{S}\bigg| \mathcal{D}\right\}. \nonumber
\end{align}
From (\ref{eq:F_d_def}) and the inequality $(1-x)^y> 1-xy$ for any
$0<x<1<y$, we have
\begin{align}
\textrm{Pr} \left\{ \mathcal{D}\right\}& = \left( 1- \left( 1-\left(\textrm{SNR}^{-(1+\gamma)}\right)^{L-1} \right)^{N_f}\right)^{KN} \nonumber\\
\label{eq:p_d_LB}& > 1- KN \left(
1-\textrm{SNR}^{-(1+\gamma)(L-1)}\right)^{N_f}.
\end{align}
Let us choose $N$ such that $N$ scales polynomially with SNR. If
$N_f$ scales faster than $\textrm{SNR}^{(1+\gamma)(L-1)}$, then the
second term of (\ref{eq:p_d_LB}) vanishes as the SNR increases,
because $\left( 1-\textrm{SNR}^{-(1+\gamma)(L-1)}\right)^{N_f} $
decreases exponentially with SNR while $N$ increases polynomially
with SNR.

Now recall that for the Grassmannian codebook approach,
${d^{[i,j]}}^2$ is bounded by ${d^{[i,j]}}^2 \le \nu_f \le
\textrm{SNR}^{-(1+\gamma)}$ along with the choice of $n_f \ge
(1+\gamma)(L-1)\log_2 \textrm{SNR}$, and that our achievability
proof is based on the upper bound on the LIF metric in
(\ref{eq:eta_tilde}). If $\mathcal{D}$ holds, then the upper bound
in (\ref{eq:eta_tilde2}) is identical to (\ref{eq:eta_tilde}), and
thus it is not difficult to show that if $N=O\left(
\textrm{SNR}^{(1+\beta)((K-1)S-L+1)}\right)$ for any
$0<\beta<\gamma$, then
\begin{equation}\label{eq:eta_RC_tend1}
\lim_{\textrm{SNR}\rightarrow \infty}\textrm{Pr} \left\{
\eta^{[i,j]}_{\textrm{RC}}\le \frac{\epsilon
\textrm{SNR}^{-1}}{KS^2}, \forall i\in \mathcal{K}, \forall j \in
\mathcal{S}\bigg| \mathcal{D}\right\} = 1,
 \end{equation}
as shown in (\ref{eq:P_prime_LB0})--(\ref{eq:P_prime_LB_final}).
From (\ref{eq:p_d_LB}) and (\ref{eq:eta_RC_tend1}), choosing the two
conditions $n_f =\log_2 N_f \ge (1+\gamma)(L-1)\log_2 \left(
\textrm{SNR}\right)$, i.e., $n_f=\Omega(L\log_2 ( \textrm{SNR}))$,
and $N=O\left( \textrm{SNR}^{(1+\beta)((K-1)S-L+1)}\right)$ for any
$0<\beta<\gamma$, the probability $\mathcal{P}_{\textrm{RC}}$ tends
to one for increasing SNR. Note that taking the limit of $N$
polynomially increasing with SNR comes merely from the strict
condition of $\mathcal{D}$. Since increasing $N$ for given $n_f$
lowers the LIF and thereby increases the achievable rate for each
selected user, $\mathcal{P}$ tends to one for any
$N=\omega\left(\textrm{SNR}^{(K-1)S-L+1}\right)$, which completes
the proof.
\end{proof}

Interestingly, Theorem \ref{theorem:RVQ} indicates that the required
$n_f$ for the random codebook is the same as that for the
Grassmannian codebook. This is an encouraging result since
analytical construction methods of the Grassmannian codebook for
large $n_f$ have been unknown, and even its numerical construction
requires excessive computational complexity.
We complete the achievability discussion by providing the following
remarks.

\begin{remark}[Random vs. Grassmannian codebook]
In the previous work on limited feedback systems, the performance
analysis has focused on the average SNR or the average rate loss
\cite{C_Au-Yeung09_TWC}. It has been known that the Grassmannian
codebook outperforms the random codebook in the average sense.
However, in our OIA framework, the focus is on the asymptotic
performance for increasing SNR, and it turns out that the asymptotic
behavior is the same for the two codebook approaches. In fact, our
result is consistent with the previous work on limited feedback
systems (see~\cite{B_Khoshnevis11_Thesis}), where the performance
gap between two codebooks was shown to be negligible as the number
of feedback bits increases.
\end{remark}

\begin{remark}[Comparison to the MIMO broadcasting channel]
For the MIMO broadcasting channel with limited feedback, where the
transmitter has $L$ antennas employing the random codebook, it was
shown in~\cite{N_Jindal06_TIT} that the achievable rate loss for
each user, denoted by $\Delta R$, coming from the finite size of the
codebook is given by $\Delta < \log_2 \left(1+\textrm{SNR} \cdot
2^{-n_f/(L-1)} \right)$. Thus, to achieve the maximum DoF for each
user, or to make the rate loss negligible as the SNR increases, the
term $\textrm{SNR} \cdot 2^{-n_f/(L-1)}$ should remain constant for
increasing SNR. That is, $n_f$ should scale faster than $(L-1)\log_2
\textrm{SNR}$. Although the system model and signaling methodology
under consideration are different from our setting, Theorems
\ref{theorem:CB} and \ref{theorem:RVQ} are consistent with this
previous result.
\end{remark}

\section{Asymptotically Optimal Receiver Design at the BSs}\label{sec:acvanced_Rx}
While using the ZF receiver is sufficient to achieve the maximum
DoF, we study the design of an enhanced receiver at the BSs in
pursuit of improving the achievable rate. Recall that each BS is not
assumed to have CSI of the cross-links from the users in the other
cells, because no coordination between BSs is assumed. In this
section, the main challenge is thus to decode the desired symbols
with no CSI of the cross-links at the receivers. For convenience,
let us rewrite the received signal at BS $i$ in (\ref{eq:y}) as
\begin{align}\label{eq:y_org_re}
\mathbf{y}_i &=  \tilde{\mathbf{H}}_i^{(c)}\mathbf{x}_i + \sum_{k=1,
k\neq i}^{K} \sum_{m=1}^{S}
\mathbf{H}_{i}^{[k,m]}\mathbf{w}^{[k,m]}x^{[k,m]} + \mathbf{z}_i,
\end{align}
where $\tilde{\mathbf{H}}_i^{(c)} \triangleq \left[
\mathbf{H}_i^{[i,1]}\mathbf{w}^{[i,1]}, \ldots
\mathbf{H}_i^{[i,S]}\mathbf{w}^{[i,S]}\right] \in \mathbb{C}^{M
\times S}$ and $\mathbf{x}_i \triangleq \left[ x^{[i,1]}, \ldots,
x^{[i,S]}\right]^{\textrm{T}}\in \mathbb{C}^{S \times 1}$. The
channel capacity $I_{\textrm{C}}$ is given by \cite{R_Blum03_JSAC}
\begin{equation} \label{eq:I_CAP}
I_{\textrm{C}} = \log_2 \det \left( \mathbf{R}_c^{-1/2}
\tilde{\mathbf{H}}_i^{(c)}\left(\tilde{\mathbf{H}}_i^{(c)}\right)^{\textrm{H}}\mathbf{R}_c^{-1/2}+
\mathbf{I}_M\right),
\end{equation}
where
\begin{equation} \label{eq:Rc}
\mathbf{R}_c = \sum_{k=1, k\neq i}^{K} \sum_{m=1}^{S}
\mathbf{H}_{i}^{[k,m]}\mathbf{w}^{[k,m]}
\left(\mathbf{H}_{i}^{[k,m]}\mathbf{w}^{[k,m]}\right)^{\textrm{H}} +
N_0 \mathbf{I}_M,
\end{equation}
which is not available at BS $i$ due to the assumption of unknown
inter-cell interfering links. The channel capacity $I_{\textrm{C}}$
is achievable with the optimal ML decoder
\begin{equation}\label{x_hat_ML}
\hat{\mathbf{x}}_i^{\textrm{ML}} = \underset{\mathbf{x}}{\arg \min}
\left( \mathbf{y}_i - \tilde{\mathbf{H}}_i^{(c)} \mathbf{x}
\right)^{\textrm{H}} \mathbf{R}_{c}^{-1} \left( \mathbf{y}_i -
\tilde{\mathbf{H}}_i^{(c)} \mathbf{x}\right),
 \end{equation}
which is infeasible to implement due to unknown $\mathbf{R}_c$.
After nulling interference by multiplying $\mathbf{U}_i$, the
received signal is given by
\begin{align}\label{eq:y_2}
\tilde{\mathbf{y}}_i &= \mathbf{U}_i^{\textrm{H}}\mathbf{y}_i =
\tilde{\mathbf{H}}_i\mathbf{x}_i + \underbrace{\sum_{k=1, k\neq
i}^{K} \sum_{m=1}^{S}
\mathbf{U}_i^{\textrm{H}}\mathbf{H}_{i}^{[k,m]}\mathbf{w}^{[k,m]}x^{[k,m]}
+ \mathbf{U}_i^{\textrm{H}}\mathbf{z}_i}_{\tilde{\mathbf{z}}_i},
\end{align}
where $\tilde{\mathbf{H}}_i \triangleq \left[
\mathbf{U}_i^{\textrm{H}}\mathbf{H}_i^{[i,1]}\mathbf{w}^{[i,1]},
\ldots
\mathbf{U}_i^{\textrm{H}}\mathbf{H}_i^{[i,S]}\mathbf{w}^{[i,S]}\right]
\in \mathbb{C}^{S \times S}$ and $\tilde{\mathbf{z}}_i \in
\mathbb{C}^{S \times 1}$ represents the effective noise.
Let us denote the covariance matrix of the effective noise after
interference nulling by
\begin{align} \label{eq:R_def}
\mathbf{R} &\triangleq E
\left\{\tilde{\mathbf{z}_i}\tilde{\mathbf{z}_i}^{\textrm{H}}\right\}
 \nonumber\\ & = \sum_{k=1, k\neq i}^{K} \sum_{m=1}^{S}
\mathbf{U}_i^{\textrm{H}}\mathbf{H}_{i}^{[k,m]}\mathbf{w}^{[k,m]}
\left(\mathbf{U}_i^{\textrm{H}}\mathbf{H}_{i}^{[k,m]}\mathbf{w}^{[k,m]}\right)^{\textrm{H}}
 + N_0 \mathbf{I}_S.
\end{align}
Then, the  ML decoder for the modified channel (\ref{eq:y_2})
becomes ${\arg \min}_\mathbf{x} \left( \tilde{\mathbf{y}}_i -
\tilde{\mathbf{H}}_i \mathbf{x} \right)^{\textrm{H}} \mathbf{R}^{-1}
\left( \tilde{\mathbf{y}}_i - \tilde{\mathbf{H}}_i \mathbf{x}
\right)$, which is also infeasible to implement since the term
$\mathbf{U}_i^{\textrm{H}}\mathbf{H}_i^{[k,m]}\mathbf{w}^{[k,m]}$
($k\in \{1, \ldots, i-1, i+1, \ldots, K\}$, $m\in \{1, \ldots, S\}$)
is not available at BS $i$.

As an alternative approach, we now introduce the following minimum
Euclidean distance receiver after interference nulling at BS $i$:
\begin{equation}\label{eq:subopt_Rx}
\hat{\mathbf{x}}_i = \arg_{\mathbf{x}} \min \left\|
\tilde{\mathbf{y}}_i - \tilde{\mathbf{H}}_i \mathbf{x} \right\|.
\end{equation}
It is worth noting that the receiver in (\ref{eq:subopt_Rx}) is not
universally optimal since $\mathbf{R}$ is not an identity matrix for
given channel instance. Now, we show the achievable rate based on
the use of the receiver in (\ref{eq:subopt_Rx}). The maximum
achievable rate of any suboptimal receiver, referred to as
\textit{mismatch capacity} \cite{A_Ganti00_TIT,N_Merhav94_TIT}, is
lower-bounded by the generalized mutual information, defined as
\cite{A_Ganti00_TIT,N_Merhav94_TIT}
\begin{equation} \label{eq:I_GMI}
I_{\textrm{GMI}} = \sup_{\theta\ge 0} I(\theta),
\end{equation}
where
\begin{equation} \label{eq:I_GMI_p}
I(\theta) \triangleq E \left[ \log_2 \frac{Q
(\tilde{\mathbf{y}}_i|\mathbf{x}_i)^{\theta}}{E \left[
Q(\tilde{\mathbf{y}}_i|\mathbf{x}_i)^{\theta}\big|\tilde{\mathbf{y}}_i,
\tilde{\mathbf{H}}_i\right]}\bigg| \tilde{\mathbf{H}}_i\right]
\end{equation}
and $Q(\tilde{\mathbf{y}}_i|\mathbf{x}_i)$ is the decoding metric
expressed in probability. The following lemma characterizes the GMI
of the decoder with mismatched noise covariance matrix.

\begin{lemma} \label{lemma:mismatch}
Consider the modified channel (\ref{eq:y_2}) and the decoding metric
with mismatched noise covariance matrix $\hat{\mathbf{R}}$, given by
\begin{equation}\label{eq:metric_mismatched_R}
Q(\tilde{\mathbf{y}}_i|\mathbf{x}_i) =
\exp\left(-(\tilde{\mathbf{y}}_i -
\tilde{\mathbf{H}}_i\mathbf{x}_i)^{\textrm{H}} \hat{\mathbf{R}}^{-1}
(\tilde{\mathbf{y}}_i - \tilde{\mathbf{H}}_i\mathbf{x}_i)\right).
\end{equation}
Then, the GMI $I_{\textrm{GMI}}$ based on the metric
(\ref{eq:metric_mismatched_R}) is given by  (\ref{eq:I_GMI}), where
$I(\theta)$ for given $\hat{\mathbf{R}}$ is expressed as
\begin{align}\label{eq:I_Lemma}
&I(\theta) = -\frac{\theta}{\log 2}
\textrm{tr}(\hat{\mathbf{R}}^{-1/2}\mathbf{R}\hat{\mathbf{R}}^{-1/2})
 + \frac{\theta}{\log2}\left[  \textrm{tr} \left(
\boldsymbol{\Omega}^{-1} \hat{\mathbf{R}}^{-1}\left(
\tilde{\mathbf{H}}_i\tilde{\mathbf{H}}_i^{\textrm{H}}+
\mathbf{R}\right)\right)\right] + \log_2 \det (\boldsymbol{\Omega}),
\end{align}
$\mathbf{R}$ is given by (\ref{eq:R_def}), and $\boldsymbol{\Omega}
\triangleq
\theta\hat{\mathbf{R}}^{-1}\tilde{\mathbf{H}}_i\tilde{\mathbf{H}}_i^{\textrm{H}}
+ \mathbf{I}_S$.
\end{lemma}

\begin{proof}
See Appendix \ref{app:modified_Rx}.
\end{proof}

We remark that if $\mathbf{R} = \hat{\mathbf{R}}$, then it is
obvious to show $I_{\textrm{GMI}}=I(\theta=1)$. In this case, using
(\ref{eq:I_Lemma}), $I(\theta=1)$ can be simplified to $I(\theta=1)
=\log_2 \det
(\mathbf{R}^{-1}\tilde{\mathbf{H}}_i\tilde{\mathbf{H}}_i^{\textrm{H}}
+ \mathbf{I}_S)$, which is equal to the channel capacity. The
following theorem characterizes the achievable rate of the proposed
minimum Euclidean distance decoder.

\begin{theorem}[Asymptotic capacity]\label{theorem:asymp_capacity}
The GMI $I_{\textrm{GMI}}$ of the codebook-based OIA using the
minimum Euclidean distance receiver in (\ref{eq:subopt_Rx}) is given
by
\begin{align} \label{eq:I_GMI_prop}
&I_{\textrm{GMI}} = \sup_{\theta\ge 0} -\frac{\theta}{\log 2}
\textrm{tr}(N_0^{-1}\mathbf{R}) + \frac{\theta}{\log2}\left[
\textrm{tr} \left( N_0^{-1}{\boldsymbol{\Omega}^{\prime}}^{-1}
\left( \tilde{\mathbf{H}}_i\tilde{\mathbf{H}}_i^{\textrm{H}}+
\mathbf{R}\right)\right)\right] + \log_2 \det
(\boldsymbol{\Omega}^{\prime}),
\end{align}
which asymptotically achieves the channel capacity $I_{\textrm{C}}$
if $N=\omega\left( \textrm{SNR}^{(K-1)S-L+1}\right)$ and
$n_f=\Omega\left(L\log_2(\textrm{SNR})\right)$, where
$\boldsymbol{\Omega}^{\prime} = \theta
N_0^{-1}\tilde{\mathbf{H}}_i\tilde{\mathbf{H}}_i^{\textrm{H}} +
\mathbf{I}_S$.
\end{theorem}

\begin{proof}
The decoder in (\ref{eq:subopt_Rx}) is equivalent to the one that
utilizes the decoding metric $Q(\tilde{\mathbf{y}}_i|\mathbf{x}_i)$
in (\ref{eq:metric_mismatched_R}) with $\hat{\mathbf{R}} = N_0
\mathbf{I}_S$. From Lemma \ref{lemma:mismatch}, the GMI based the
minimum Euclidean distance decoder is thus given by
(\ref{eq:I_GMI_prop}).

Now, we show that as $N$ increases, $I_{\textrm{GMI}}$ approaches
$I_{\textrm{C}}$ with increasing SNR. Recall that the user selection
and weight vector design are performed such that interference is
aligned to the reference basis matrix $\mathbf{Q}_i\in \mathbb{C}^{M
\times (M-S)}$ as much as possible. From Theorem \ref{theorem:CB}
and \ref{theorem:RVQ}, if $N=\omega\left(
\textrm{SNR}^{(K-1)S-L+1}\right)$ and
$n_f=\Omega\left(L\log_2(\textrm{SNR})\right)$, then the sum of
received interference aligned to $\mathbf{U}_i$ can be made
arbitrarily small with high probability, thereby resulting in
\begin{equation}\label{eq:interf_aligned}
\mathbf{H}_i^{[k,m]} \mathbf{w}^{[k,m]} \in \textrm{span}
(\mathbf{Q}_i)
\end{equation}
from the fact that the interference term of (\ref{eq:y_org_re}) is
given by $\sum_{k=1, k\neq i}^{K}
\sum_{m=1}^{S}\mathbf{H}_{i}^{[k,m]}\mathbf{w}^{[k,m]}x^{[k,m]}$.
Since
$\mathbf{U}_i^{\textrm{H}}\mathbf{H}_{i}^{[k,m]}\mathbf{w}^{[k,m]} =
0$, the interference term of (\ref{eq:y_2}) is canceled out, and
thus it follows that $\mathbf{R} \rightarrow \hat{\mathbf{R}} = N_0
\mathbf{I}_S$. In this case, we have
\begin{equation} \label{eq:I_GMI2}
I_{\textrm{GMI}} \rightarrow I_{\textrm{GMI}}^{*}= \log_2 \det
\left( N_0^{-1}
\tilde{\mathbf{H}}_i\tilde{\mathbf{H}}_i^{\textrm{H}}+
\mathbf{I}_S\right).
\end{equation}

Now we prove that $I_{\textrm{GMI}}^{*}$ in (\ref{eq:I_GMI2})
asymptotically achieves $I_{\textrm{C}}$. Since $[\mathbf{U}_i,
\mathbf{Q}_i]\in \mathbb{C}^{M \times M}$ is an orthogonal matrix,
$I_{\textrm{GMI}}^{*}$ can be rewritten as
\begin{align}
I_{\textrm{GMI}}^{*} & = \log_2 \det \left( N_0^{-1} \mathbf{U}_i^{\textrm{H}}\tilde{\mathbf{H}}_i^{(c)}\left(\tilde{\mathbf{H}}_i^{(c)}\right)^{\textrm{H}}\mathbf{U}_i + \mathbf{I}_S\right), \nonumber\\
& = \log_2 \det \left( \left[
                       \begin{array}{cc}
                         {N_0}^{-1}\mathbf{U}_i^{\textrm{H}}\tilde{\mathbf{H}}_i^{(c)}\left(\tilde{\mathbf{H}}_i^{(c)}\right)^{\textrm{H}}\mathbf{U}_i  & \mathbf{0} \nonumber\\
                         \mathbf{0} & \mathbf{0} \nonumber\\
                       \end{array}
                     \right]\!+\! \mathbf{I}_M\!
\right)\nonumber\\
& = \log_2 \det \Biggl( [\mathbf{U}_i, \mathbf{Q}_i] \cdot \left[
                       \begin{array}{cc}
                         {N_0}^{-1}\mathbf{U}_i^{\textrm{H}}\tilde{\mathbf{H}}_i^{(c)}\left(\tilde{\mathbf{H}}_i^{(c)}\right)^{\textrm{H}}\mathbf{U}_i  & \mathbf{0} \\
                         \mathbf{0} & \mathbf{0}\nonumber \\
                       \end{array}
                     \right] \cdot[\mathbf{U}_i, \mathbf{Q}_i] ^{\textrm{H}}+ \mathbf{I}_M
\Biggr)\nonumber\\
\label{eq:I_GMI_3}& = \log_2 \det \left(
\mathbf{D}\tilde{\mathbf{H}}_i^{(c)}\left(\tilde{\mathbf{H}}_i^{(c)}\right)^{\textrm{H}}\mathbf{D}
+ \mathbf{I}_M\right),
\end{align}
where
\begin{align}
\mathbf{D} & = [\mathbf{U}_i, \mathbf{Q}_i] \left[ \begin{array}{cc}
                                                     N_0^{-1/2}\mathbf{I}_S & \mathbf{0} \\
                                                     \mathbf{0} & \mathbf{0}
                                                   \end{array}
\right][\mathbf{U}_i, \mathbf{Q}_i]^{\textrm{H}}  = N_0^{-1/2}
\mathbf{U}_i {\mathbf{U}_i}^{\textrm{H}}. \nonumber
\end{align}
From the fact that from (\ref{eq:interf_aligned}),
$\mathbf{H}_i^{[k,m]} \mathbf{w}^{[k,m]}$ can be represented as a
linear combination of $\mathbf{q}_{i,1}, \ldots, \mathbf{q}_{i,
M-S}$, $\mathbf{R}_c$  in (\ref{eq:Rc}) is written as
\begin{align}
\mathbf{R}_c & = \sum_{k=1, k\neq i}^{K} \sum_{m=1}^{S}  \sum_{l=1}^{M-S} \sum_{t=1}^{M-S}\tilde{\xi}^{[k,m]}_{l,t} \mathbf{q}_{i,l}\mathbf{q}_{i,t}^{\textrm{H}}+ N_0 \mathbf{I}_M \nonumber\\
& = \sum_{l=1}^{M-S}\sum_{t=1}^{M-S}\xi_{l,t} \mathbf{q}_{i,l}\mathbf{q}_{i,t}^{\textrm{H}}+ N_0 \mathbf{I}_M \nonumber\\
& = \sum_{l=1}^{M-S}\sum_{t=1}^{M-S}\xi_{l,t} \mathbf{q}_{i,l}\mathbf{q}_{i,t}^{\textrm{H}}+ N_0 [\mathbf{Q}_i, \mathbf{U}_i] \cdot [\mathbf{Q}_i, \mathbf{U}_i]^{\textrm{H}}\nonumber\\
& = [\mathbf{Q}_i, \mathbf{U}_i] \cdot \left[
                        \begin{array}{cc}
                          \boldsymbol{\Xi} & \mathbf{0} \\
                          \mathbf{0} & N_0 \mathbf{I}_S \\
                        \end{array}
                      \right]\cdot [\mathbf{Q}_i,
                      \mathbf{U}_i]^{\textrm{H}}, \nonumber
\end{align}
where $\tilde{\xi}^{[k,m]}_{l,t}$ denotes the component coefficient
of $\mathbf{H}^{[k,m]}_i\mathbf{w}^{[k,m]}$ on
$\mathbf{q}_{i,l}\mathbf{q}_{i,t}^{\textrm{H}}$, $\xi_{l,t} =
\sum_{k=1, k\neq i}^{K} \sum_{m=1}^{S}  \tilde{\xi}^{[k,m]}_{l,t}$,
and
\begin{equation}
\boldsymbol{\Xi} = \left[ \begin{array}{cccc}
                            \xi_{1,1}+N_0 & \xi_{1,2} & \cdots & \xi_{1,M-S} \\
                            \xi_{2,1} & \xi_{2,2}+N_0 & \cdots & \xi_{2,M-S} \\
                            \vdots & \ddots & \vdots & \vdots \\
                            \xi_{M-S,1} & \xi_{M-S,2} & \cdots & \xi_{M-S,M-S}+N_0
                          \end{array}
\right]. \nonumber
\end{equation}
Since each coefficient $\xi_{l,t}$ is chosen from a continuous
distribution, $\boldsymbol{\Xi}$ has full rank almost surely, and
thus is invertible. Therefore, we get
\begin{equation}
\mathbf{R}_c^{-1/2} = [\mathbf{Q}_i, \mathbf{U}_i]  \left[
                        \begin{array}{cc}
                          \boldsymbol{\Xi}^{-1/2} & \mathbf{0} \\
                          \mathbf{0} & {N_0}^{-1/2} \mathbf{I}_S \\
                        \end{array}
                      \right][\mathbf{Q}_i,
                      \mathbf{U}_i]^{\textrm{H}}. \nonumber
\end{equation}
For given channel instance, as $N_0$ decreases,
$\mathbf{R}_c^{-1/2}$ becomes
\begin{equation} \label{eq:R_c_asympt}
\mathbf{R}_c^{-1/2} \rightarrow {N_0}^{-1/2} \mathbf{U}_i
{\mathbf{U}_i}^{\textrm{H}}.
\end{equation}
Applying the asymptotic $\mathbf{R}_c^{-1/2}$ of
(\ref{eq:R_c_asympt}) to (\ref{eq:I_CAP}) finally yields
$I_{\textrm{C}}$, which is equal to $I_{\textrm{GMI}}^{*}$ of
(\ref{eq:I_GMI_3}). This completes the proof.
\end{proof}

As shown in Theorem \ref{theorem:asymp_capacity}, the minimum
Euclidean distance receiver asymptotically achieves the channel
capacity even without any coordination between the BSs or users.
However, it is worth noting that if the interference alignment level
is too low due to small $N$ or $n_f$ to satisfy the conditions in
Theorems \ref{theorem:CB} and \ref{theorem:RVQ}, then the achievable
rate in (\ref{eq:I_GMI_prop}) may be lower than that based on the ZF
receiver. Thus, in small $N$ and $n_f$ regimes, there may exist
crossovers, where the achievable rate of the two schemes is
switched, which will be shown in Section \ref{SEC:Sim} via numerical
evaluation. We conclude our discussion on the receiver design with
the following remark.

\begin{remark} [DoF achievability of the optimal receiver]
Even with the use of the ML receiver in (\ref{x_hat_ML}) based on
full knowledge of $\mathbf{R}_{c}$, the user and $n_f$ scaling
conditions to achieve $KS$ DoF are the same as those based of the ZF
receiver case, which make the amount of interference bounded even
for increasing SNR.
\end{remark}

\section{Numerical Results} \label{SEC:Sim}

In this section, we run computer simulations to verify the
performance of the proposed two types of codebook-based OIA schemes,
i.e., the Grassmannian and random codebook-based OIAs, for finite
system parameters SNR, $N$, and $n_f$. For comparison, the max-SNR
scheme was used, in which each user employs eigen-beamforming in
terms of maximizing its received SNR and the belonging BS selects
the $S$ users who have the SNR values up to the $S$th largest one.
The SVD-based OIA scheme in~\cite{H_Yang13_TWC} is also compared as
a baseline method.

Figures \ref{fig:LIF_N} show a log-log plot of the sum-LIF (or
equivalently, the sum of generating interference) versus $N$ for the
MIMO IMAC with $K=2$, $M=3$, $L=2$, and (a) $S=2$ or (b) $S=3$.
From Theorems \ref{theorem:CB} and \ref{theorem:RVQ}, the system parameter $S$ governs the trade-off between the achievable DoF and the number of users required to guarantee such DoF~\cite{B_Jung11_TC, H_Yang13_TWC}.\footnote{While the sum-LIF with $S=1$ is lowest compared to the cases with $S=2$ and $S=3$, the case with $S=1$ provides the smallest achievable DoF. For more discussions about optimizing $S$, we refer to \cite{H_Yang13_TWC}.} It is seen that the sum-LIF increases as $S$ grows for any given scheme. However, as addressed in\cite{B_Jung11_TC, H_Yang13_TWC}, note that a smaller LIF does not necessarily leads to a higher achievable rate, especially in the high SNR regime. 
In addition, the trade-off between $n_f$ and the sum-LIF is clearly
seen from Fig. \ref{fig:LIF_N}. Although the sum-LIF level of the
codebook-based OIA scheme decreases as $n_f$ increases, its
decreasing rate of the sum-LIF with respect to $N$, representing the
slope of the sum-LIF curve, slightly differs from that of the
SVD-based OIA unless $n_f$ increases according to increasing $N$.

\begin{figure}
\begin{center}
\subfigure[]{
\includegraphics[width=0.6\textwidth]{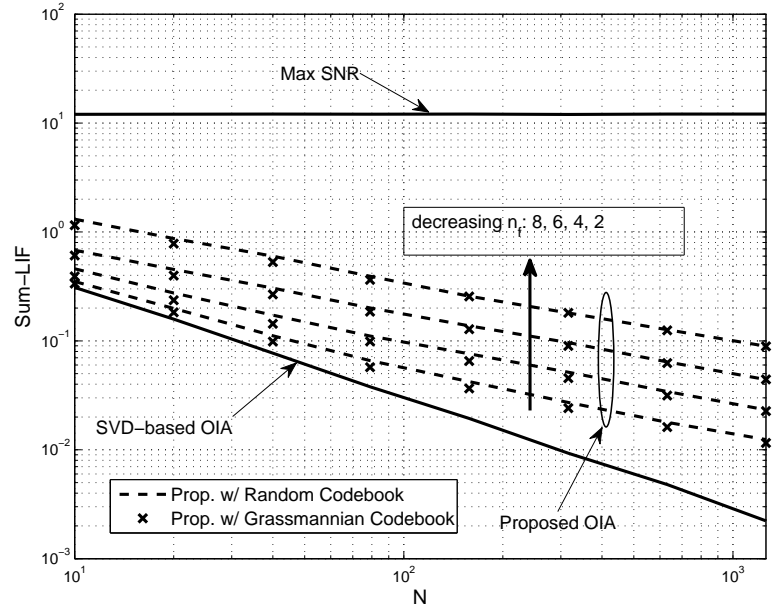}\label{fig:LIF_N_S2} } \hspace{-20pt}
     \subfigure[]{
\includegraphics[width=0.6\textwidth]{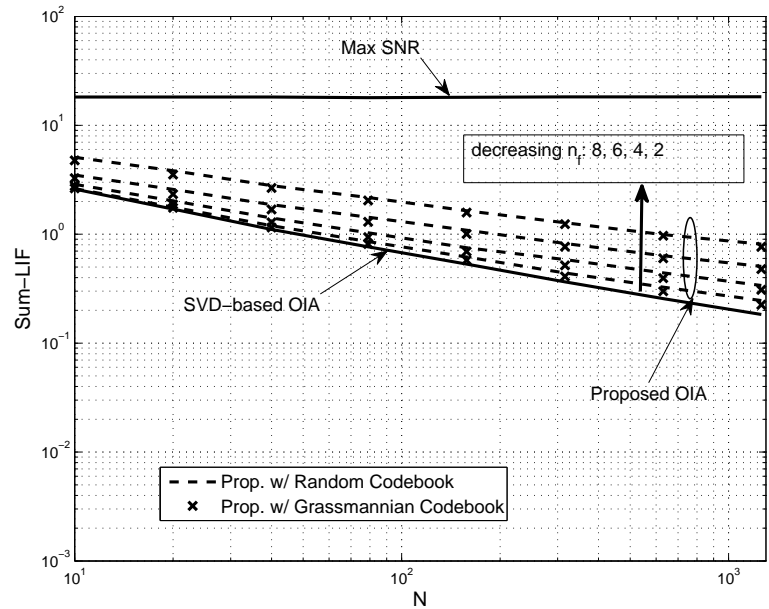}
  \label{fig:LIF_N_S3}}
  \caption{The sum-LIF versus $N$ when $K=2$, $M=3$, and $L=2$. (a) $S=2$. (b) $S=3$.} \label{fig:LIF_N}
\end{center}
\end{figure}

\begin{figure}
\begin{center}
  \includegraphics[width=0.6\textwidth]{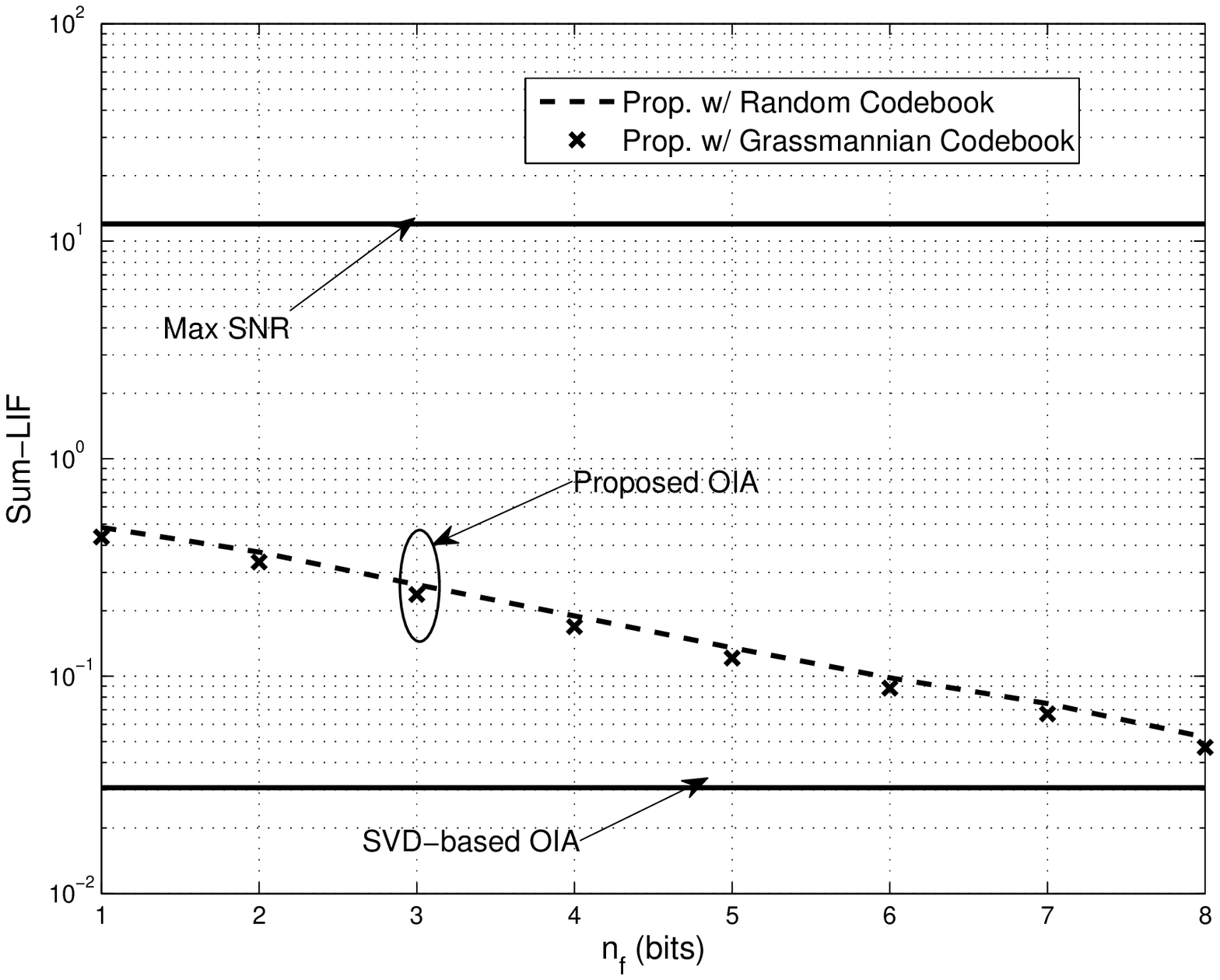}\\
  \caption{The sum-LIF versus $n_f$ when $K=2$, $M=3$, $L=2$, $S=2$, and $N=100$.}\label{fig:LIF_nf}
  \end{center}
\end{figure}

Figure \ref{fig:LIF_nf} illustrates a linear-log plot of the sum-LIF
versus $n_f$ when $K=2$, $M=3$, $L=2$, $S=2$ and $N=100$. As
expected from Theorems \ref{theorem:CB} and \ref{theorem:RVQ}, it is
seen that the decreasing rate of the sum-LIF is almost the same for
both codebook-based OIA schemes. As $n_f$ increases, even with
finite $N$, the sum-LIF level for both codebook-based OIA schemes
becomes close to that for the SVD-based OIA.

\begin{figure}
\begin{center}
\subfigure[]{
\includegraphics[width=0.6\textwidth]{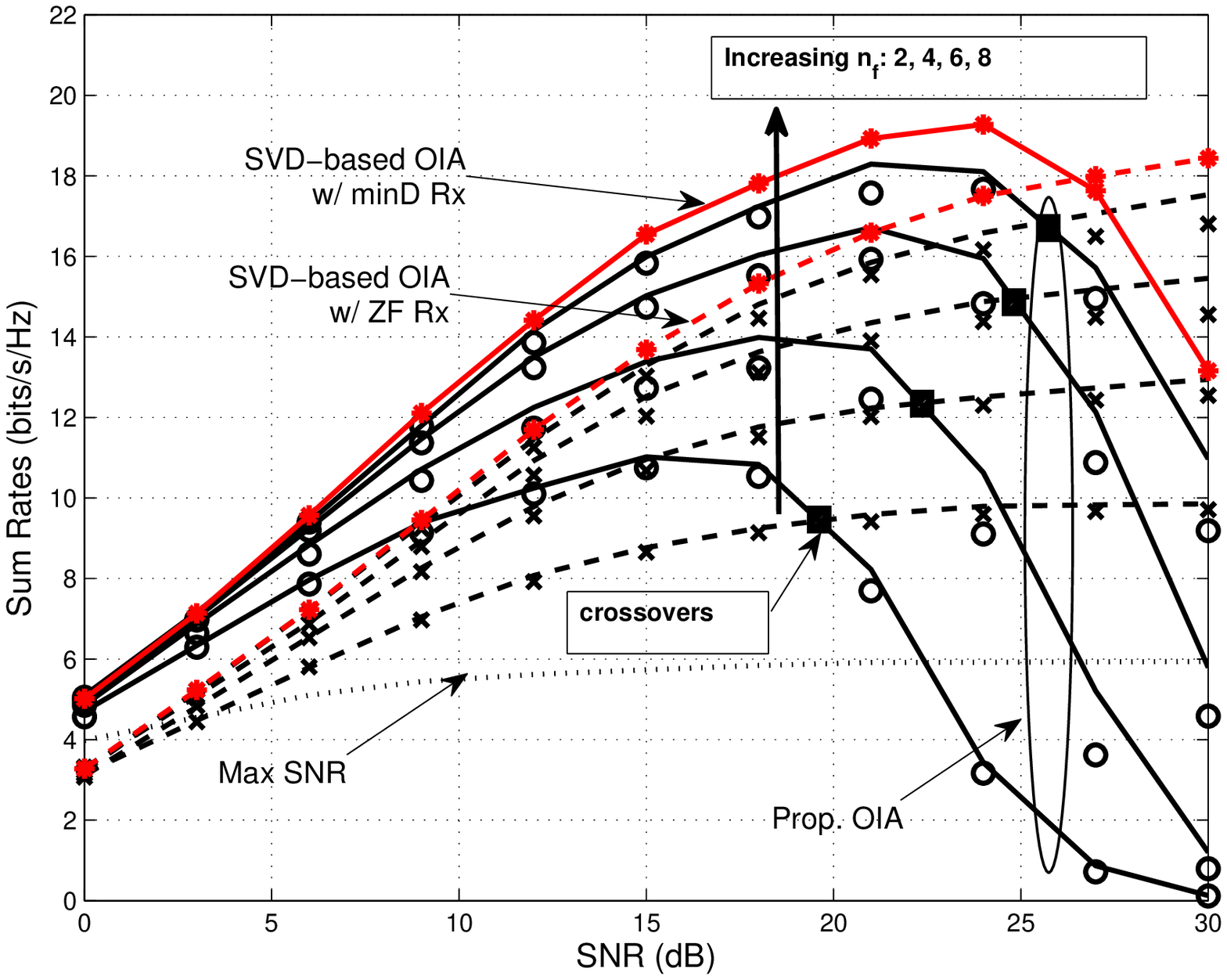}\label{fig:rate_SNR_N20} }
     \subfigure[]{
\includegraphics[width=0.6\textwidth]{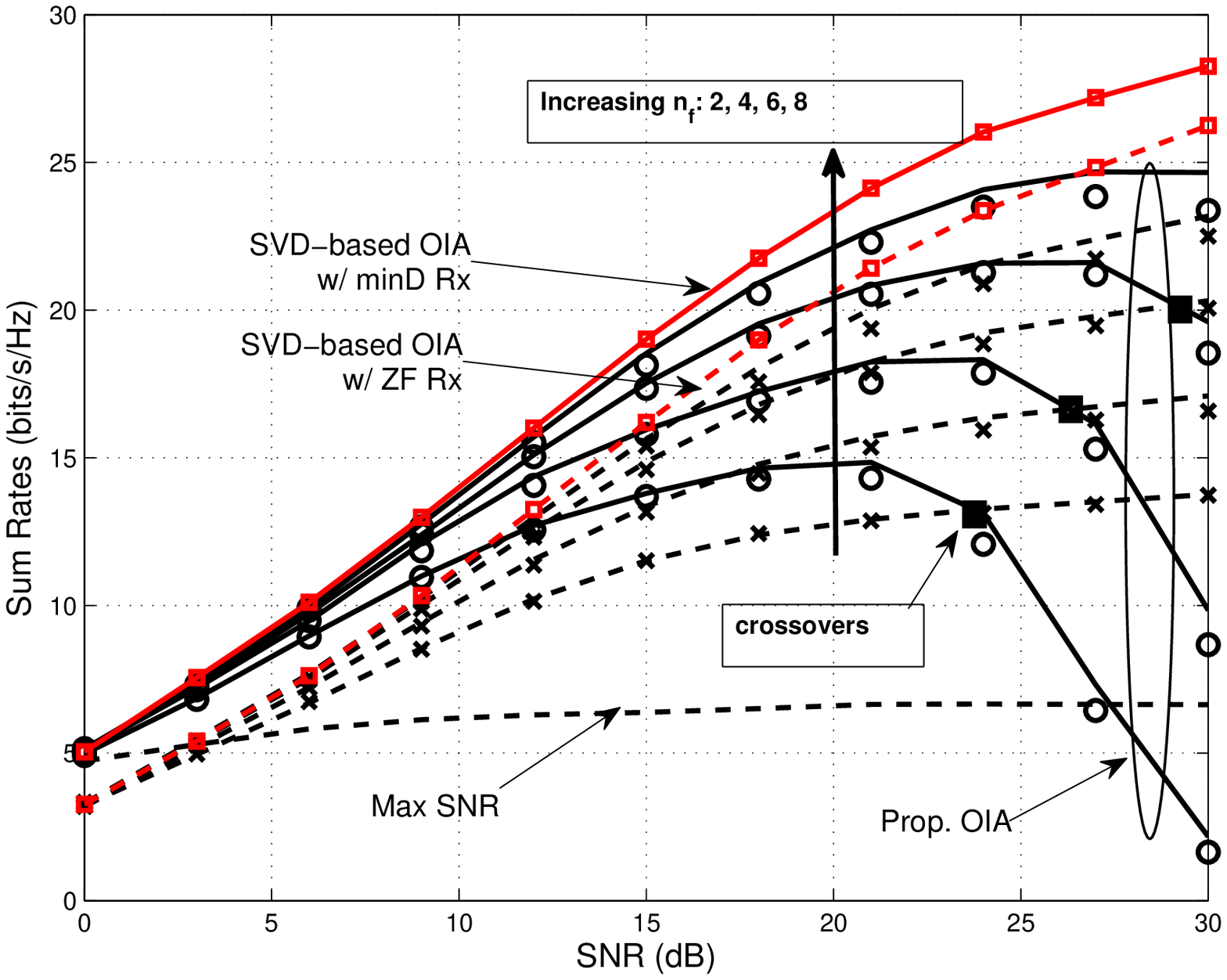}
  \label{fig:rate_SNR_N100}}
  \caption{The achievable rate versus SNR when $K=2$, $M=3$, $L=2$, and $S=2$. (a) $N=20$. (b)
  $N=100$.}\label{fig:rate_SNR}
\end{center}
\end{figure}

Figures \ref{fig:rate_SNR} depicts the achievable rate versus SNR
when $K=2$, $M=3$, $L=2$, $S=2$, and (a) $N=20$ or (b) $N=100$. We
consider the following four receiver structures for the proposed
codebook-based OIA:
\begin{itemize}
\item Scheme 1: ZF receiver with the Grassmannian codebook (dashed line)
\item Scheme 2: ZF receiver with the random codebook (x)
\item Scheme 3: minimum Euclidean distance receiver with the Grassmannian codebook (solid line)
\item Scheme 4: minimum Euclidean distance receiver with the random codebook (o)
\end{itemize}
A relationship between the sum-rate for given $N$ and the number of
feedforward bits, $n_f$, is observed. It is first seen that as
$n_f=8$, the proposed codebook-based OIA schemes closely obtain the
achievable rate of the SVD-based OIA. It is also seen that the gain
coming from the Grassmannian codebook over the the random codebook
is marginal. From Theorem \ref{theorem:asymp_capacity}, we remark
that the achievable rate based on the minimum Euclidean distance
receiver asymptotically achieves the channel capacity if
interference $\mathbf{H}_{i}^{[k,m]}\mathbf{w}^{[k,m]}$ is perfectly
aligned to $\mathbf{Q}_i$ at BS $i$; that is, the covariance matrix
of interference in (\ref{eq:R_def}), $\sum_{k=1, k\neq i}^{K}
\sum_{m=1}^{S}
\mathbf{U}_i^{\textrm{H}}\mathbf{H}_{i}^{[k,m]}\mathbf{w}^{[k,m]}
\left(\mathbf{U}_i^{\textrm{H}}\mathbf{H}_{i}^{[k,m]}\mathbf{w}^{[k,m]}\right)^{\textrm{H}}$,
becomes negligible compared to $N_0 \mathbf{I}_S$ due to the fact
that interference is sufficiently aligned for large $N$. In
addition, it is observed that in the low to mid SNR regimes, using
the minimum Euclidean distance receiver leads to a higher sum rate
than the ZF receiver case even for practical $N$. However, as the
SNR increases beyond a certain point, i.e., in the high SNR regime,
the covariance matrix of interference becomes dominant, thus
yielding a performance degradation of the minimum Euclidean distance
receiver. On the other hand, since the ZF receiver has no such
limitation, its achievable rate increases with SNR. In consequence,
for given $N$, there exist crossovers, where the achievable rate of
the two schemes is switched. It is furthermore seen that when $N$
increases, these crossovers appear at higher SNRs, because our
system is less affected by the covariance matrix of interference
owing to a better interference alignment.

\begin{figure}
\begin{center}
  \includegraphics[width=0.6\textwidth]{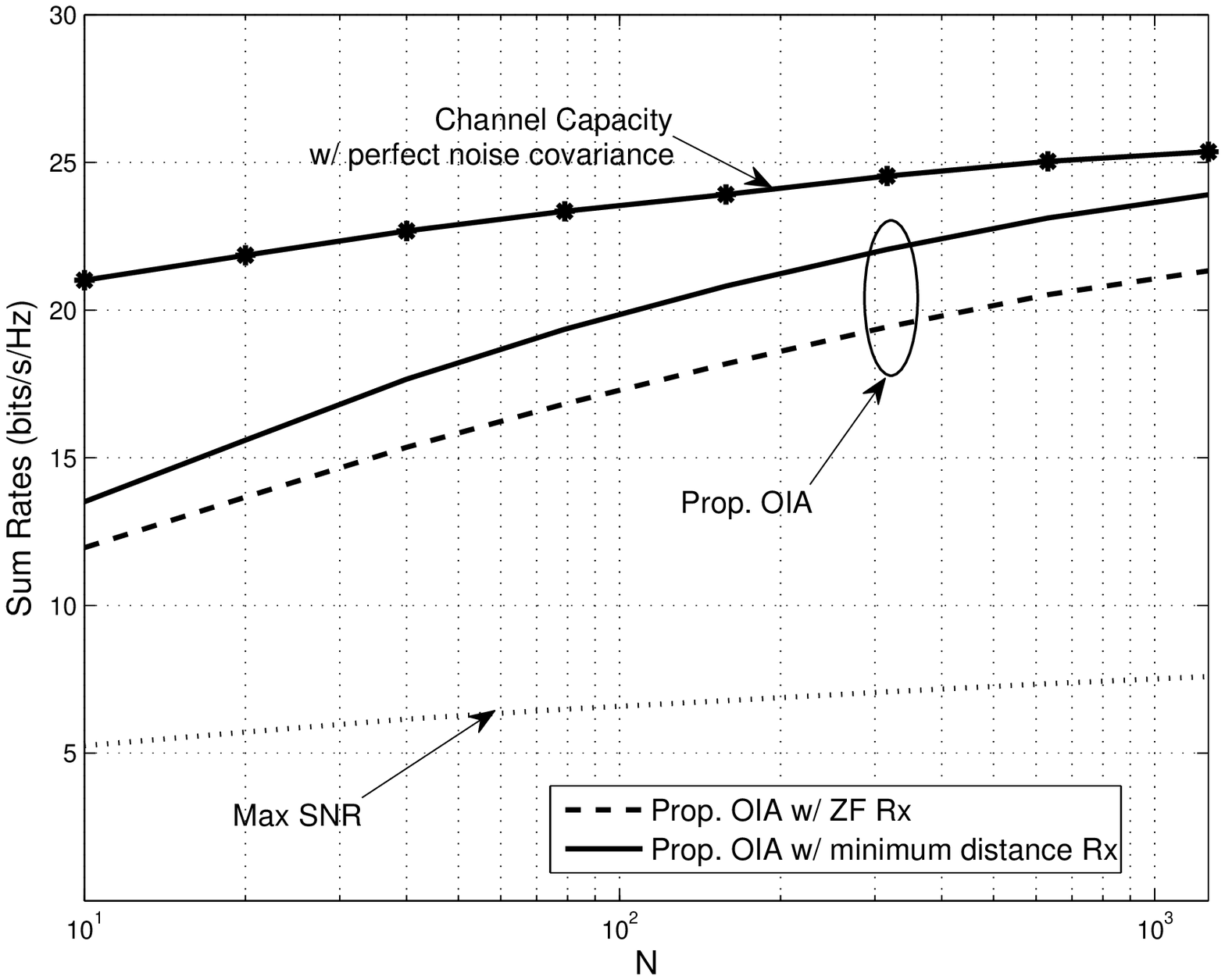}\\
  \caption{The achievable rate versus $N$ when $K=2$, $M=3$, $L=2$, $S=2$, SNR=20dB, and $n_f=6$. The random codebook is used.}\label{fig:rate_N}
  \end{center}
\end{figure}

Figure \ref{fig:rate_N} shows a log-linear plot of the achievable
rate versus $N$ when $K=2$, $M=3$, $L=2$, $S=2$, SNR=20dB, and
$n_f=6$ when the random codebook is used. As shown in Theorem
\ref{theorem:asymp_capacity}, it is seen that the GMI of the
codebook-based OIA using the minimum Euclidean distance receiver
asymptotically achieves the channel capacity as $N$ increases. On
the other hand, the achievable rate of the codebook-based OIA using
the ZF receiver exhibits a constant gap even in large $N$ regime,
compared to that of the minimum Euclidean distance receiver. This
observation is consistent with previous results on the single-user
MIMO channel, showing that there exists a constant SNR gap between
the channel capacity and the achievable rate based on the ZF
receiver in the high SNR regime.

\section{Conclusion} \label{SEC:Conc}
For the MIMO IMAC, we have proposed two different types of
codebook-based OIA methods and analyzed the codebook size required
to achieve the same user scaling condition and DoF as the SVD-based
OIA case.
We have shown that the required codebook size scaling is the same
for both of the random and Grassmannian codebooks.
In addition, we have shown that the simple minimum Euclidean
distance receiver operating even with no CSI of inter-cell
interfering links achieves the channel capacity as $N$ increases.
Numerical examples have shown that it suffices for finite $n_f$ to
almost obtain the achievable rate of the SVD-based OIA, e.g., the
case where $n_f=8$ and $L=2$, and that the minimum Euclidean
distance receiver enhances the achievable rate based on the ZF
receiver especially in the low to mid SNR regimes. \pagebreak[0]

\appendices
\section{Proof of (\ref{eq:p_c_asympt})} \label{app:p_c}
By using $\nu_f \le \textrm{SNR}^{-(1+\gamma)}$ in
(\ref{eq:N_f_inequality}), $p_c$ defined in
(\ref{eq:P_prime_LB1_00}) can be lower-bounded by
\begin{align}
p_c &\ge \textrm{Pr} \left\{ {\sigma^{[i,j]}_L}^2 \ge (1+\delta) {\sigma^{[i,j]}_1}^2\textrm{SNR}^{-(1+\gamma)}, \forall i\in \mathcal{K}, j\in \mathcal{N}\right\}\nonumber \\
\label{eq:p_c_LB2}& = \left(\textrm{Pr} \left\{
\frac{{\sigma^{[i,j]}_1}^2}{{\sigma^{[i,j]}_L}^2} \le
(1+\delta)^{-1}\textrm{SNR}^{1+\gamma}\right\}\right)^{KN},
\end{align}
where (\ref{eq:p_c_LB2}) comes from the fact that
$\mathbf{G}^{[i,j]}$ is independent for different $i$ or $j$, and
thus their singular values are independent for different users. Note
that $\frac{{\sigma^{[i,j]}_1}^2}{{\sigma^{[i,j]}_L}^2}$ is the
condition number of
${\mathbf{G}^{[i,j]}}^{\textrm{H}}\mathbf{G}^{[i,j]}$. At this
point, we introduce the following lemma on the CDF of the condition
number.

\begin{lemma} \label{lemma:CDF_cond}
The CDF of $\frac{{\sigma^{[i,j]}_1}^2}{{\sigma^{[i,j]}_L}^2}$,
denoted by $F_c(x)$, is lower-bounded by
\begin{equation} \label{eq:F_c_def}
F_c(x) \ge 1-\rho x^{-((K-1)S-L+1)},
\end{equation}
where $\rho$ is a constant determined by $K$, $S$, and $L$.
\end{lemma}

\begin{proof}
Since each channel coefficient is assumed to be chosen from a
continuous distribution, $\mathbf{G}^{[i,j]}$ has full rank almost
surely \cite{A_Edelman89_PhD}.  Moreover, assuming that $(K-1)S>L$,
${\mathbf{G}^{[i,j]}}^{\textrm{H}} \mathbf{G}^{[i,j]}$ is the
full-rank central Wishart matrix, i.e.,
${\mathbf{G}^{[i,j]}}^{\textrm{H}} \mathbf{G}^{[i,j]} \sim
\mathcal{CW}_{L}\left((K-1)S, \mathbf{I}_L\right)$. Using the
high-tail distribution of the complementary CDF in \cite[Theorem
4]{M_Matthaiou10_TC}, the CDF is bounded by (\ref{eq:F_c_def}),
where $ \rho = \left( \sum_{l=1}^{L} \kappa_l \mu_l\right)$. Here,
$\kappa_l = \frac{(-1)^{l(l-1)/2}}{((K-1)S-L+1)!}$ and $\mu_l =
\int_{0}^{\infty} \lambda_1^{((K-1)S-L+1)} g_{L-1, (K-1)S}
(\lambda_1) d\lambda_1$, where $\lambda_1 \triangleq
{\sigma^{[i,j]}_1}^2$ and $g_{L-1, (K-1)S} (\cdot)$ denotes the
distribution of the largest eigenvalue of a reduced Wishart matrix
$\tilde{\mathbf{G}} \sim \mathcal{CW}_{L-1}\left((K-1)S,
\mathbf{I}_{L-1}\right)$. Thus, $\rho$ is determined only by $K$,
$S$, and $L$.
\end{proof}

Now, from Lemma \ref{lemma:CDF_cond}, (\ref{eq:p_c_LB2}) is further
bounded by
\begin{align}
p_c &\!\ge\! \left( 1\!-\!\rho (1+\delta)^{(K-1)S+L-1}\textrm{SNR}^{-(1+\gamma)((K-1)S-L+1)}  \right)^{KN} \nonumber\\
\label{eq:p_c_LB_APP}&\ge 1\!-\!\rho KN
(1+\delta)^{(K-1)S+L-1}\textrm{SNR}^{-(1+\gamma)((K-1)S-L+1)},
\end{align}
where (\ref{eq:p_c_LB_APP}) follows from $(1-x)^y> 1-xy$ for any
$0<x<1<y$. Therefore, if $N$ scales slower than
$\textrm{SNR}^{(1+\gamma)((K-1)S-L+1)}$ for given  $\gamma$, i.e.,
$N=O\left(\textrm{SNR}^{(1+\beta)((K-1)S-L+1)}\right)$ where
$\beta<\gamma$, then $\lim_{\textrm{SNR}\rightarrow
\infty}p_c\rightarrow 1$ for any given $\delta$, which proves the
argument (\ref{eq:p_c_asympt}).

\section{Proof of Lemma \ref{lemma:mismatch}}\label{app:modified_Rx}

Let us first consider the numerator of the logarithmic term in
(\ref{eq:I_GMI_p}) as follows.
\begin{align}
E\left[ \log_2  Q (\tilde{\mathbf{y}}_i|\mathbf{x}_i)^{\theta}
|\tilde{\mathbf{H}}_i \right]
&= -\frac{\theta}{\log 2} E \left[ (\tilde{\mathbf{y}}_i - \tilde{\mathbf{H}}_i\mathbf{x}_i)^{\textrm{H}} \hat{\mathbf{R}}^{-1} (\tilde{\mathbf{y}}_i - \tilde{\mathbf{H}}_i\mathbf{x}_i)\big|\tilde{\mathbf{H}}_i\right] \nonumber\\
& =  -\frac{\theta}{\log 2} E \left[ \tilde{\mathbf{z}}_i^{\textrm{H}} \hat{\mathbf{R}}^{-1} \tilde{\mathbf{z}}_i\big|\tilde{\mathbf{H}}_i\right] \nonumber\\
\label{eq:num_final_p}& =-\frac{\theta}{\log 2}
\textrm{tr}(\hat{\mathbf{R}}^{-1/2}\mathbf{R}\hat{\mathbf{R}}^{-1/2}),
\end{align}
where (\ref{eq:num_final_p}) follows from $E[\mathbf{z}^{\textrm{H}}
\mathbf{A} \mathbf{z} ] = \textrm{tr}(\mathbf{A}^{1/2}
\mathbf{R}\mathbf{A}^{1/2})$ for any random vector $\mathbf{z}$ with
$E\{\mathbf{z}\mathbf{z}^{\textrm{H}}\} = \mathbf{R}$ and
$E\{\mathbf{z}\} = \mathbf{0}$ and for any conjugate symmetric
matrix $\mathbf{A}$.

Let us turn to the denominator of the logarithmic term in
(\ref{eq:I_GMI_p}). Here, given that $\tilde{\mathbf{H}}_i$ and
$\tilde{\mathbf{y}}_i$ are deterministic, the expectation is taken
over $\mathbf{x}_i$. Since $\mathbf{x}_i$ is an $S$-dimensional
complex Gaussian random vector, i.e., the probability density
function of $\mathbf{x}_i$ is given by $1/\pi^{S} \exp\left(
-\|\mathbf{x}_i\|^2\right)$, we have
\begin{align}
&E \left[ Q(\tilde{\mathbf{y}}_i|\mathbf{x}_i)^{\theta}\big|\tilde{\mathbf{y}}_i, \tilde{\mathbf{H}}_i\right] \nonumber\\ &= \int \exp\left( -\theta(\tilde{\mathbf{y}}_i - \tilde{\mathbf{H}}_i\mathbf{x}_i)^{\textrm{H}} \hat{\mathbf{R}}^{-1} (\tilde{\mathbf{y}}_i - \tilde{\mathbf{H}}_i\mathbf{x}_i)\right) \frac{1}{\pi^{S}} \exp\left( -\|\mathbf{x}_i\|^2\right) d\mathbf{x}_i \nonumber\\
& = \frac{1}{\pi^{S}} \int \exp\bigl( -\theta(\tilde{\mathbf{y}}_i - \tilde{\mathbf{H}}_i\mathbf{x}_i)^{\textrm{H}} \hat{\mathbf{R}}^{-1} (\tilde{\mathbf{y}}_i - \tilde{\mathbf{H}}_i\mathbf{x}_i)-\|\mathbf{x}_i\|^2\bigr)d\mathbf{x}_i \nonumber\\
\label{eq:denom_p}& = \frac{1}{\pi^{S}} \int \exp\left(
-A\right)d\mathbf{x}_i.
\end{align}
Here, $A$ can be further expressed as
\begin{align}
A &\triangleq \theta (\tilde{\mathbf{y}}_i - \tilde{\mathbf{H}}_i\mathbf{x}_i)^{\textrm{H}} \hat{\mathbf{R}}^{-1} (\tilde{\mathbf{y}}_i - \tilde{\mathbf{H}}_i \mathbf{x}_i)+\|\mathbf{x}_i\|^2 \nonumber \\
\label{eq:A2_p}&= \mathbf{x}_i^{\textrm{H}}(
\theta\tilde{\mathbf{H}}_i^{\textrm{H}}\hat{\mathbf{R}}^{-1}\tilde{\mathbf{H}}_i
+ \mathbf{I}_S)\mathbf{x}_i -\theta\tilde{\mathbf{y}}_i^{\textrm{H}}
\hat{\mathbf{R}}^{-1}\tilde{\mathbf{H}}_i\mathbf{x}_i - \theta
\mathbf{x}_i^{\textrm{H}}\tilde{\mathbf{H}}_i^{\textrm{H}}\hat{\mathbf{R}}^{-1}
\tilde{\mathbf{y}}_i + \theta
\tilde{\mathbf{y}}_i^{\textrm{H}}\hat{\mathbf{R}}^{-1}\tilde{\mathbf{y}}_i.
\end{align}
Letting $\tilde{\boldsymbol{\Omega}}=
\theta\tilde{\mathbf{H}}_i^{\textrm{H}}\hat{\mathbf{R}}^{-1}\tilde{\mathbf{H}}_i
+ \mathbf{I}_S$, it follows that
\begin{align}
\label{eq:A3_p}A &= \left\| \tilde{\boldsymbol{\Omega}}^{1/2}
\mathbf{x}_i - \mathbf{C}\tilde{\mathbf{y}}_i \right\|^2 +
\theta\tilde{\mathbf{y}}_i^{\textrm{H}}\hat{\mathbf{R}}^{-1}\tilde{\mathbf{y}}_i-
\tilde{\mathbf{y}}_i^{\textrm{H}} \mathbf{C}^{\textrm{H}} \mathbf{C}
\tilde{\mathbf{y}}_i,
\end{align}
where $\mathbf{C}$ is given by $\mathbf{C} =
\theta\tilde{\boldsymbol{\Omega}}^{-1/2}
\tilde{\mathbf{H}}_i^H\hat{\mathbf{R}}^{-1}$, which comes from the
equivalence of (\ref{eq:A2_p}) and (\ref{eq:A3_p}). Now let us
further simplify the last two terms of (\ref{eq:A3_p}) as
\begin{align}
\theta\tilde{\mathbf{y}}_i^{\textrm{H}}\hat{\mathbf{R}}^{-1}\tilde{\mathbf{y}}_i - \tilde{\mathbf{y}}_i^{\textrm{H}} \mathbf{C}^{\textrm{H}} \mathbf{C} \tilde{\mathbf{y}}_i  & = \theta \tilde{\mathbf{y}}_i^{\textrm{H}}\left( \hat{\mathbf{R}}^{-1} - \theta \hat{\mathbf{R}}^{-1}\tilde{\mathbf{H}}_i \tilde{\boldsymbol{\Omega}}^{-1} \tilde{\mathbf{H}}_i^{\textrm{H}}\hat{\mathbf{R}}^{-1} \right) \tilde{\mathbf{y}}_i \nonumber\\
&= \theta
\tilde{\mathbf{y}}_i^{\textrm{H}}\hat{\mathbf{R}}^{-1/2}\left(
\mathbf{I}_S - \theta \hat{\mathbf{R}}^{-1/2}\tilde{\mathbf{H}}_i
\tilde{\boldsymbol{\Omega}}^{-1}
\tilde{\mathbf{H}}_i^{\textrm{H}}\hat{\mathbf{R}}^{-1/2} \right)
\hat{\mathbf{R}}^{-1/2} \tilde{\mathbf{y}}_i. \nonumber
\end{align}
Without loss of generality, it follows that
$\hat{\mathbf{R}}^{-1/2}\tilde{\mathbf{H}}_i = \boldsymbol{\Phi}
\boldsymbol{\Lambda} \mathbf{T}^{\textrm{H}}$, where
$\boldsymbol{\Phi}\in \mathbb{C}^{S \times S}$ and $\mathbf{T}\in
\mathbb{C}^{S \times S}$ are orthogonal matrices and
$\boldsymbol{\Lambda}$ is an $(S\times S)$-dimensional diagonal
matrix. Then, we get
\begin{align}
\tilde{\boldsymbol{\Omega}} =
\theta\tilde{\mathbf{H}}_i^{\textrm{H}}\hat{\mathbf{R}}^{-1}\tilde{\mathbf{H}}_i
+ \mathbf{I}_S  = \mathbf{T}\left( \theta \boldsymbol{\Lambda}^2 +
\mathbf{I}_S \right) \mathbf{T}^{\textrm{H}}.
\label{eq:tilde_sigma_p}
\end{align}
Inserting (\ref{eq:tilde_sigma_p}) to
$\hat{\mathbf{R}}^{-1/2}\tilde{\mathbf{H}}_i
\tilde{\boldsymbol{\Omega}}^{-1}
\tilde{\mathbf{H}}_i^{\textrm{H}}\hat{\mathbf{R}}^{-1/2}$  gives us
\begin{align}
\hat{\mathbf{R}}^{-1/2}\tilde{\mathbf{H}}_i
\tilde{\boldsymbol{\Omega}}^{-1}
\tilde{\mathbf{H}}_i^{\textrm{H}}\hat{\mathbf{R}}^{-1/2}  &=
\boldsymbol{\Phi} \boldsymbol{\Lambda}^2 \left( \theta
\boldsymbol{\Lambda}^2 + \mathbf{I}_S \right)^{-1}
\boldsymbol{\Phi}^{\textrm{H}}, \nonumber
\end{align}
which yields
\begin{align}
\theta\tilde{\mathbf{y}}_i^{\textrm{H}}\hat{\mathbf{R}}^{-1}\tilde{\mathbf{y}}_i
- \tilde{\mathbf{y}}_i^{\textrm{H}} \mathbf{C}^{\textrm{H}}
\mathbf{C} \tilde{\mathbf{y}}_i
 &=  \theta \tilde{\mathbf{y}}_i^{\textrm{H}}\hat{\mathbf{R}}^{-1/2}\boldsymbol{\Phi}\left( \mathbf{I}_S - \theta \boldsymbol{\Lambda}^2 \left( \theta \boldsymbol{\Lambda}^2 + \mathbf{I}_S \right)^{-1}  \right) \boldsymbol{\Phi}^{\textrm{H}}\hat{\mathbf{R}}^{-1/2}\tilde{\mathbf{y}}_i \nonumber\\
\label{eq:const_term2_p} &= \theta \tilde{\mathbf{y}}_i^{\textrm{H}}\hat{\mathbf{R}}^{-1/2}\boldsymbol{\Phi} \left( \theta \boldsymbol{\Lambda}^2 + \mathbf{I}_S \right)^{-1}  \boldsymbol{\Phi}^{\textrm{H}}\hat{\mathbf{R}}^{-1/2}\tilde{\mathbf{y}}_i \\
& = \theta \tilde{\mathbf{y}}_i^{\textrm{H}}\hat{\mathbf{R}}^{-1/2} \left( \theta\hat{\mathbf{R}}^{-1/2}\tilde{\mathbf{H}}_i\tilde{\mathbf{H}}_i^{\textrm{H}}\hat{\mathbf{R}}^{-1/2} + \mathbf{I}_S \right)^{-1}\hat{\mathbf{R}}^{-1/2}\tilde{\mathbf{y}}_i \nonumber\\
& = \theta \tilde{\mathbf{y}}_i^{\textrm{H}}\left( \theta\tilde{\mathbf{H}}_i\tilde{\mathbf{H}}_i^{\textrm{H}} + \hat{\mathbf{R}} \right)^{-1}\tilde{\mathbf{y}}_i \nonumber\\
\label{eq:const_term5_p}& = \theta
\tilde{\mathbf{y}}_i^{\textrm{H}}\boldsymbol{\Omega}^{-1}\hat{\mathbf{R}}^{-1}\tilde{\mathbf{y}}_i,
\end{align}
where (\ref{eq:const_term2_p}) follows immediately from evaluating
the diagonal terms, and (\ref{eq:const_term5_p}) follows from
$\boldsymbol{\Omega} \triangleq
\theta\hat{\mathbf{R}}^{-1}\tilde{\mathbf{H}}_i\tilde{\mathbf{H}}_i^{\textrm{H}}
+ \mathbf{I}_S$. Inserting (\ref{eq:const_term5_p}) and
(\ref{eq:A3_p}) to (\ref{eq:denom_p}) gives us\pagebreak[0]
\begin{align}
& E \left[ Q (\tilde{\mathbf{y}}_i|\mathbf{x}_i)^{\theta}|\tilde{\mathbf{y}}_i, \tilde{\mathbf{H}}_i\right] \nonumber\\
&= \frac{1}{\pi^{S}} \int \exp\left(\!-\! \left\| \tilde{\boldsymbol{\Omega}}^{1/2} \mathbf{x}_i \!-\! \mathbf{C}\tilde{\mathbf{y}}_i \right\|^2 \!-\!  \theta \tilde{\mathbf{y}}_i^{\textrm{H}} \boldsymbol{\Omega}^{-1}\hat{\mathbf{R}}^{-1}\tilde{\mathbf{y}}_i\!\right) d\mathbf{x}_i \nonumber\\
& = \frac{1}{\pi^{S}} \exp \left( -  \theta \tilde{\mathbf{y}}_i^{\textrm{H}} \boldsymbol{\Omega}^{-1}\hat{\mathbf{R}}^{-1}\tilde{\mathbf{y}}_i\right) \int \exp\left(- \left\| \tilde{\boldsymbol{\Omega}}^{1/2} \mathbf{x}_i - \mathbf{C}\tilde{\mathbf{y}}_i \right\|^2 \right) d\mathbf{x}_i \nonumber\\
\label{eq:denom_final15_p}&= \frac{1}{\pi^{S}} \exp \left( -  \theta \tilde{\mathbf{y}}_i^{\textrm{H}} \boldsymbol{\Omega}^{-1}\hat{\mathbf{R}}^{-1}\tilde{\mathbf{y}}_i\right)  \pi^{S} \det \left(\tilde{\boldsymbol{\Omega}}^{-1}\right) \\
\label{eq:denom_fianl2_p}& = \exp \left( -  \theta
\tilde{\mathbf{y}}_i^{\textrm{H}}
\boldsymbol{\Omega}^{-1}\hat{\mathbf{R}}^{-1}\tilde{\mathbf{y}}_i\right)
\det \left(\boldsymbol{\Omega}^{-1}\right),
\end{align}
where (\ref{eq:denom_final15_p}) follows from the fact that for
$\mathbf{x}, \mathbf{m}\in \mathbb{C}^{S \times 1}$ and conjugate
symmetric $\mathbf{A}, \mathbf{B}\in \mathbb{C}^{S \times S}$,
\begin{align}
&\int \exp \left(- (\mathbf{A}\mathbf{x}-\mathbf{m})^{\textrm{H}}\mathbf{B}^{-1} (\mathbf{A}\mathbf{x}-\mathbf{m}) \right) d\mathbf{x} \nonumber\\ &= \int \exp \left(-(\mathbf{x}-\mathbf{A}^{-1}\mathbf{m})^{\textrm{H}}\mathbf{A}^{\textrm{H}}\mathbf{B}^{-1}\mathbf{A} (\mathbf{x}-\mathbf{A}^{-1}\mathbf{m}) \right) d\mathbf{x}\nonumber \\
&=\pi^{S}
\det\left(\mathbf{A}^{-1}\mathbf{B}(\mathbf{A}^{\textrm{H}})^{-1}\right),
\nonumber
\end{align}
and (\ref{eq:denom_fianl2_p}) follows from $\det(
\tilde{\boldsymbol{\Omega}} ) =\det( \boldsymbol{\Omega} )$.
Inserting (\ref{eq:num_final_p}) and (\ref{eq:denom_fianl2_p}) to
(\ref{eq:I_GMI_p}) gives us
\begin{align}
I(\theta) &= -\frac{\theta}{\log 2}  \textrm{tr}(\hat{\mathbf{R}}^{-1/2}\mathbf{R}\hat{\mathbf{R}}^{-1/2}) - E \left[ \log_2 \left(\exp \left( -  \theta \tilde{\mathbf{y}}_i^{\textrm{H}} \boldsymbol{\Omega}^{-1}\hat{\mathbf{R}}^{-1}\tilde{\mathbf{y}}_i\right) \det \left(\boldsymbol{\Omega}^{-1}\right) \right)\big|\tilde{\mathbf{H}}_i\right] \nonumber \\
 & = -\frac{\theta}{\log 2} \textrm{tr}(\hat{\mathbf{R}}^{-1/2}\mathbf{R}\hat{\mathbf{R}}^{-1/2})  + \frac{\theta}{\log2}E\left[\tilde{\mathbf{y}}_i^{\textrm{H}} \boldsymbol{\Omega}^{-1}\hat{\mathbf{R}}^{-1}\tilde{\mathbf{y}}_i\big| \tilde{\mathbf{H}}_i \right] + \log_2 \det
 (\boldsymbol{\Omega}). \nonumber
\end{align}
From
\begin{align}
E\left[\tilde{\mathbf{y}}_i^{\textrm{H}} \boldsymbol{\Omega}^{-1}\hat{\mathbf{R}}^{-1}\tilde{\mathbf{y}}_i\big| \tilde{\mathbf{H}}_i \right]  & = E \left[ \left( \mathbf{x}_i^{\textrm{H}} \tilde{\mathbf{H}}_i^{\textrm{H}}+\mathbf{z}_i^{\textrm{H}}\right)\boldsymbol{\boldsymbol{\Omega}}^{-1}\hat{\mathbf{R}}^{-1}\left( \tilde{\mathbf{H}}_i \mathbf{x}_i+ \mathbf{z}_i \right)\right] \nonumber\\
& = \textrm{tr} \left( \boldsymbol{\Omega}^{-1}
\hat{\mathbf{R}}^{-1}\tilde{\mathbf{H}}_i\tilde{\mathbf{H}}_i^{\textrm{H}}
\right) + \textrm{tr} \left(
\boldsymbol{\Omega}^{-1}\hat{\mathbf{R}}^{-1}\mathbf{R}\right),
\nonumber
\end{align}
we finally have
\begin{align}
I(\theta) &= -\frac{\theta}{\log 2}
\textrm{tr}(\hat{\mathbf{R}}^{-1/2}\mathbf{R}\hat{\mathbf{R}}^{-1/2})
+ \frac{\theta}{\log2}\bigl[  \textrm{tr} \left(
\boldsymbol{\Omega}^{-1}
\hat{\mathbf{R}}^{-1}\tilde{\mathbf{H}}_i\tilde{\mathbf{H}}_i^{\textrm{H}}
\right) + \textrm{tr} \left(
\boldsymbol{\Omega}^{-1}\hat{\mathbf{R}}^{-1}\mathbf{R}\right)\bigr]
+ \log_2 \det (\boldsymbol{\Omega}), \nonumber
\end{align}
which completes the proof of the lemma.



\begin{thebibliography}{10}
\providecommand{\url}[1]{#1} \csname url@rmstyle\endcsname
\providecommand{\newblock}{\relax} \providecommand{\bibinfo}[2]{#2}
\providecommand\BIBentrySTDinterwordspacing{\spaceskip=0pt\relax}
\providecommand\BIBentryALTinterwordstretchfactor{4}
\providecommand\BIBentryALTinterwordspacing{\spaceskip=\fontdimen2\font
plus \BIBentryALTinterwordstretchfactor\fontdimen3\font minus
  \fontdimen4\font\relax}
\providecommand\BIBforeignlanguage[2]{{%
\expandafter\ifx\csname l@#1\endcsname\relax
\typeout{** WARNING: IEEEtran.bst: No hyphenation pattern has been}%
\typeout{** loaded for the language `#1'. Using the pattern for}%
\typeout{** the default language instead.}%
\else \language=\csname l@#1\endcsname \fi #2}}

\bibitem{V_Cadambe08_TIT}
V.~R. Cadambe and S.~A. Jafar, ``Interference alignment and degrees
of freedom
  of the \uppercase{K}-user interference channel,'' \emph{IEEE Trans. Inf.
  Theory}, vol.~54, no.~8, pp. 3425--3441, Aug. 2008.

\bibitem{M_MaddahAli08_TIT}
M.~A. Maddah-Ali, A.~S. Motahari, and A.~K. Khandani,
``Communication over
  \uppercase{MIMO X} channels: Interference alignment, decomposition, and
  performance analysis,'' \emph{IEEE Trans. Inf. Theory}, vol.~54, no.~8, pp.
  3457--3470, Aug. 2008.

\bibitem{B_Jung11_CL}
B.~C. Jung and W.-Y. Shin, ``Opportunistic interference alignment
for
  interference-limited cellular \uppercase{TDD} uplink,'' \emph{IEEE Commun.
  Lett.}, vol.~15, no.~2, pp. 148--150, Feb. 2011.

\bibitem{B_Jung11_TC}
B.~C. Jung, D.~Park, and W.-Y. Shin, ``Opportunistic interference
mitigation
  achieves optimal degrees-of-freedom in wireless multi-cell uplink networks,''
  \emph{IEEE Trans. Commun.}, vol.~60, no.~7, pp. 1935--1944, Jul. 2012.

\bibitem{C_Suh08_Allerton}
C.~Suh and D.~Tse, ``Interference alignment for cellular networks,''
in
  \emph{Proc. 46th Annual Allerton Conf. Communication, Control, and
  Computing}, Urbana-Champaign, IL, Sept. 2008, pp. 1037 -- 1044.

\bibitem{K_Comadam08_GLOBECOM}
K.~Gomadam, V.~R. Cadambe, and S.~A. Jafar, ``A distributed
numerical approach
  to interference alignment and applications to wireless interference
  networks,'' \emph{IEEE Trans. Inf. Theory}, vol.~57, no.~6, pp. 3309--3322,
  June 2011.

\bibitem{H_Yang13_TWC}
H.~J. Yang, W.-Y. Shin, B.~C. Jung, and A.~Paulraj, ``Opportunistic
  interference alignment of \uppercase{MIMO} interfering multiple-access
  channels,'' \emph{IEEE Trans. Wireless Commun.}, vol.~12, no.~5, pp.
  2180--2192, May 2013.

\bibitem{H_Yang12_ISIT}
------, ``Opportunistic interference alignment of \uppercase{MIMO IMAC} :
  Effect of user scaling over degrees-of-freedom,'' in \emph{Proc. IEEE Int'l
  Symp. Inf. Theory (ISIT)}, Cambridge, MA, July 2012, pp. 2646--2650.

\bibitem{J_Thukral09_ISIT}
J.~Thukral and H.~B\"{o}lcskei, ``Interference alignment with
limited
  feedback,'' in \emph{Proc. IEEE Int'l Symp. Inf. Theory (ISIT)}, Seoul,
  Korea, July 2009, pp. 1759--1763.

\bibitem{R_Krishnamachari10_ISIT}
R.~T. Krishnamachari and M.~K. Varanasi, ``Interference alignment
under limited
  feedback for \uppercase{MIMO} interference channels,'' in \emph{Proc. IEEE
  Int'l Symp. Inf. Theory (ISIT)}, Austin, TX, June 2010, pp. 619--623.

\bibitem{B_Mondal06_TSP}
B.~Mondal and R.~W. Heath, Jr., ``Performance analysis of quantized
beamforming
  \uppercase{MIMO} systems,'' \emph{IEEE Trans. Signal Process.}, vol.~54,
  no.~12, pp. 4753--4766, Dec. 2006.

\bibitem{T_Yoo07_JSAC}
T.~Yoo, N.~Jindal, and A.~Goldsmith, ``Multi-antenna downlink
channels with
  limited feedback and user selection,'' \emph{IEEE J. Select. Areas Commun.},
  vol.~25, no.~7, pp. 1478--1491, Sept. 2007.

\bibitem{N_Jindal06_TIT}
N.~Jindal, ``\uppercase{MIMO} broadcast channels with finite-rate
feedback,''
  \emph{IEEE Trans. Inf. Theory}, vol.~52, no.~11, pp. 5045--5060, Nov. 2006.

\bibitem{TS36.213}
\emph{TS 36.213, Evolved Universal Terrestrial Radio Access
  (\uppercase{E-UTRA}); Physical layer procedures}, 3GPP Std., v.11.2.0.

\bibitem{L_Choi04_TWC}
L.~Choi and R.~D. Murch, ``A transmit preprocessing technique for
multiuser
  \uppercase{MIMO} systems using a decomposition approach,'' \emph{IEEE Trans.
  Wireless Commun.}, vol.~3, no.~1, pp. 20--24, Jan. 2004.

\bibitem{Z_Pan04_TWC}
Z.~Pan, K.-K. Wong, and T.-S. Ng, ``Generalized multiuser orthogonal
  space-division multiplexing,'' \emph{IEEE Trans. Wireless Commun.}, vol.~3,
  no.~6, pp. 1969--1973, Nov. 2004.

\bibitem{TS36.212}
\emph{TS 36.212, Evolved Universal Terrestrial Radio Access
  (\uppercase{E-UTRA}); Multiplexing and channel coding}, 3GPP Std., v.11.2.0.

\bibitem{C_Chae08_TSP}
C.-B. Chae, D.~Mazzarese, T.~Inoue, and R.~W. Heath, Jr.,
``Coordinated
  beamforming for the multiuser \uppercase{MIMO} broadcast channel with limited
  feedforward,'' \emph{IEEE Trans. Signal Process.}, vol.~56, no.~12, pp.
  6044--6056, Dec. 2008.

\bibitem{I_Hwang10_USPatent}
``Apparatus and method for beamforming with limited feedforward
channel in
  multiple input multiple output wireless communication system,'' US Patent
  7786934, Aug. 2010.

\bibitem{J_Jose11_TVT}
J.~Jose, A.~Ashikhmin, P.~Whiting, and S.~Vishwanath, ``Channel
estimation and
  linear precoding in multiuser multiple-antenna \uppercase{TDD} systems,''
  \emph{IEEE Trans. Veh. Technol.}, vol.~60, no.~5, pp. 2102--2116, June 2011.

\bibitem{D_Samardzija06_CISS}
D.~Samardzija, L.~Xiao, and N.~Mandayam, ``Impact of pilot assisted
channel
  state estimation on multiple antenna multiuser \uppercase{TDD} systems with
  spatial filtering,'' in \emph{Proc. 40th Annual Conf. Information Sciences
  and Systems}, Lucent Technol. Bell Labs, Holmdel, NJ, Mar. 2006, pp.
  381--385.

\bibitem{J_Jose11_TWC}
J.~Jose, A.~Ashikhmin, T.~L. Marzetta, and S.~Vishwanath, ``Pilot
contamination
  and precoding in multi-cell \uppercase{TDD} systems,'' \emph{IEEE Trans.
  Wireless Commun.}, vol.~10, no.~8, pp. 2640--2651, Aug. 2011.

\bibitem{L_Equigua11_SPL}
L.~Soriano-Equigua, J.~S\'{a}nchez-Garc\'{\i}a, J.~Flores-Troncoso,
and R.~W.
  Heath, Jr., ``Noniterative coordinated beamforming for multiuser
  \uppercase{MIMO} systems with limited feedforward,'' \emph{IEEE Signal
  Process. Lett.}, vol.~18, no.~12, pp. 701--704, Dec. 2011.

\bibitem{D_Love03_TIT}
D.~J. Love and R.~W. Heath, Jr., ``Grassmannian beamforming for
multiple-input
  multple-output wireless systems,'' \emph{IEEE Trans. Inf. Theory}, vol.~49,
  no.~10, pp. 2735--2747, Oct. 2003.

\bibitem{W_Dai08_TIT}
W.~Dai, Y.~E. Liu, and B.~Rider, ``Quantization bounds on
\uppercase{G}rassmann
  manifolds and applications to \uppercase{MIMO} communications,'' \emph{IEEE
  Trans. Inf. Theory}, vol.~54, no.~3, pp. 1108--1123, Mar. 2008.

\bibitem{A_Barg02_TIT}
A.~Barg and D.~Y. Nogin, ``Bounds on packings of spheresin the
  \uppercase{G}rassmann manifold,'' \emph{IEEE Trans. Inf. Theory}, vol.~48,
  no.~9, pp. 2450--2454, Sept. 2002.

\bibitem{S_Jin08_TC}
S.~Jin, M.~R. McKay, X.~Gao, and I.~B. Collings, ``\uppercase{MIMO}
  multichannel beamforming: \uppercase{SER} and outage using new eigenvalue
  distributions of complex noncentral \uppercase{W}ishart matrices,''
  \emph{IEEE Trans. Commun.}, vol.~56, no.~3, pp. 424--434, Mar. 2008.

\bibitem{C_Au-Yeung09_TWC}
C.~K. Au-Yeung and D.~J. Love, ``Optimization and tradeoff analysis
of two-way
  limited feedback beamforming systems,'' \emph{IEEE Trans. Wireless Commun.},
  vol.~8, no.~5, pp. 2570--2579, May 2009.

\bibitem{B_Khoshnevis11_Thesis}
B.~Khoshnevis, ``Multiple-antenna communications with limited
channel state
  information,'' Ph.D. dissertation, University of Toronto, 2011.

\bibitem{R_Blum03_JSAC}
R.~S. Blum, ``\uppercase{MIMO} capacity with interference,''
\emph{IEEE J.
  Selec. Area. Commun.}, vol.~21, no.~5, pp. 793--801, June 2003.

\bibitem{A_Ganti00_TIT}
A.~Ganti, A.~Lapidoth, and I.~E. Telatar, ``Mismatched decoding
revisited:
  General alphabets, channels with memory, and the wide-band limit,''
  \emph{IEEE Trans. Inf. Theory}, vol.~46, no.~7, pp. 2315--2328, Nov. 2000.

\bibitem{N_Merhav94_TIT}
N.~Merhav, G.~Kaplan, A.~Lapidoth, and S.~Shamai~(Shitz), ``On
information
  rates for mismatched decoders,'' \emph{IEEE Trans. Inf. Theory}, vol.~40,
  no.~6, pp. 1953--1967, Nov. 1994.

\bibitem{A_Edelman89_PhD}
A.~Edelman, ``Eigenvalues and condition numbers of random
matrices,'' Ph.D.
  dissertation, Messachusetts Institute of Technology, 1989.

\bibitem{M_Matthaiou10_TC}
M.~Matthaiou, M.~R. McKay, P.~J. Smith, and J.~A. Nossek, ``On the
condition
  number distribution of complex \uppercase{W}ishart matrices,'' \emph{IEEE
  Trans. Commun.}, vol.~58, no.~6, pp. 1705--1717, June 2010.

\end{thebibliography}
\end{document}